\documentclass[11pt]{article}
\usepackage{amssymb}
\usepackage{amsmath}
\usepackage{amsthm}
\usepackage{latexsym}
\usepackage{graphicx}
\usepackage{lineno}
\newtheorem{theo}{Theorem}[section]
\newtheorem{prop}[theo]{Proposition}
\newtheorem{cor}[theo]{Corollary}
\newtheorem{lemma}[theo]{Lemma}
\theoremstyle{definition}
\newtheorem{defi}[theo]{Definition}
\newtheorem{exa}[theo]{Example}
\newtheorem{rem}[theo]{Remark}
\newtheorem{question}[theo]{Question}
\newtheorem{problem}[theo]{Problem}

\numberwithin{equation}{section}

\newcommand{\N}{{\mathbb N}}
\newcommand{\F}{{\mathbb F}}

\newcommand{\Z}{{\mathbb Z}}
\newcommand{\Q}{{\mathbb Q}}
\newcommand{\C}{{\mathbb C}}

\newcommand{\cA}{{\mathcal A}}
\newcommand{\cC}{{\mathcal C}}
\newcommand{\cG}{{\mathcal G}}
\newcommand{\cH}{{\mathcal H}}

\newcommand{\cP}{{\mathcal P}}

\newcommand{\wcP}{\widehat{\mathcal P}}
\newcommand{\wwcP}{\widehat{\phantom{\big|}\hspace*{.5em}}\hspace*{-.9em}\wcP}

\newcommand{\wcPchil}{\wcP^{^{\scriptscriptstyle[\chi,l]}}}
\newcommand{\wcPchir}{\wcP^{^{\scriptscriptstyle[\chi,r]}}}
\newcommand{\cQ}{{\mathcal Q}}
\newcommand{\cB}{{\mathcal B}}

\newcommand{\cL}{{\mathcal L}}

\newcommand{\cR}{{\mathcal R}}

\newcommand{\cN}{{\mathcal N}}
\newcommand{\cU}{{\mathcal U}}

\newcommand{\bP}{{\mathbf P}}
\newcommand{\bH}{{\mathbf H}}

\newcommand{\GL}{\mbox{\rm GL}}
\newcommand{\widesim}[1][1.5]{\scalebox{#1}[1]{$\sim$}}
\newcommand{\rank}{\mbox{\rm rank}\,}
\renewcommand{\Vec}{\mbox{\rm vec}}

\newcommand{\im}{\mbox{\rm rowspan}\,}

\newcommand{\wt}{{\rm wt}}
\newcommand{\wtRT}{{\rm wt}_{\rm RT}}
\newcommand{\wtRK}{{\rm wt}_{\rm rk}}
\newcommand{\wtRKq}{{\rm wt}_{{\rm rk},\,q}}
\newcommand{\wtNRT}{{\rm wt}_{\rm NRT}}
\newcommand{\wtH}{{\rm wt}_{\rm H}}
\newcommand{\LT}{\mbox{\rm LT}}

\newcommand{\T}{\mbox{$\!^{\sf T}$}}
\newcommand{\inner}[1]{\mbox{$\langle{#1}\rangle$}}
\newcommand{\Mon}{\mathrm{Mon}}
\newcommand{\MonF}{\mathrm{Mon}(n,\F)}
\newcommand{\MonR}{\mathrm{Mon}(n,R)}
\newcommand{\MonUR}{\mathrm{Mon}_U(n,R)}
\newcommand{\glnr}{\mathrm{GL}(n,R)}
\newcommand{\supp}{\mathrm{supp}}
\newcommand{\ltnr}{\mathrm{LT}(n,R)}
\newcommand{\ideal}[1]{\mbox{$\langle{#1}]$}}
\newcommand{\Smallfourmat}[4]{\mbox{$\left(\begin{smallmatrix}{#1}&{#2}\\{#3}&{#4}\end{smallmatrix}\right)$}}
\newcommand{\Smallsixmat}[6]{\mbox{$\left(\begin{smallmatrix}{#1}&{#2}&{#3}\\{#4}&{#5}&{#6}\end{smallmatrix}\right)$}}

\newcounter{alp}
\newcounter{ara}
\newcounter{rom}
\newenvironment{romanlist}{\begin{list}{(\roman{rom})\hfill}{\usecounter{rom}
     \topsep0ex \labelwidth.7cm \leftmargin.7cm \labelsep0cm
     \rightmargin0cm \parsep0ex \itemsep.4ex
     \partopsep1ex}}{\end{list}}
\newenvironment{alphalist}{\begin{list}{(\alph{alp})\hfill}{\usecounter{alp}
     \topsep.5ex \labelwidth.6cm \leftmargin.6cm \labelsep0cm
     \rightmargin0cm \parsep0ex \itemsep0ex
     \partopsep1.6ex}}{\end{list}}
\newenvironment{arabiclist}{\begin{list}{(\arabic{ara})\hfill}{\usecounter{ara}
     \topsep0ex \labelwidth.6cm \leftmargin.6cm \labelsep0cm
     \rightmargin0cm \parsep0ex \itemsep0ex
     \partopsep1.6ex}}{\end{list}}

\begin{document}
%\linenumbers

\title{MacWilliams Extension Theorems and the Local-Global Property for Codes over Frobenius Rings}
\date\today
\author{Aleams Barra\thanks{Faculty of Mathematics and Natural Sciences, Bandung Institute of Technology,
Bandung 40132, Indonesia (email: barra@math.itb.ac.id). A. Barra acknowledges support by the ICF-1.321 International Conference Grant from IMHERE ITB.}\quad and
Heide Gluesing-Luerssen\thanks{Department of Mathematics, University of Kentucky, Lexington KY 40506-0027, USA
(email: heide.gl@uky.edu). H. Gluesing-Luerssen was partially supported by the National Science Foundation Grants
DMS-0908379 and DMS-1210061.}
}

\maketitle

%%%%%%%%%%%%%%%%%%%%%%%%%%%%%%%%%%%%%%%%%%%%%%%%%%%%%%%

{\bf Abstract:}
The MacWilliams extension theorem is investigated for various weight functions over finite Frobenius rings.
The problem is reformulated in terms of a local-global property for subgroups of the general linear group.
Among other things, it is shown that the extension theorem holds true for poset weights if and only if the underlying poset
is hierarchical.
Specifically, the Rosenbloom-Tsfasman weight for vector codes satisfies the extension theorem, whereas
the Niederreiter-Rosenbloom-Tsfasman weight for matrix codes does not.
A short character-theoretic proof of the well-known MacWilliams extension theorem for the homogeneous weight is provided.
Moreover it is shown that the extension theorem carries over to direct products of weights, but not to symmetrized products.

\smallskip
{\bf Keywords:}
MacWilliams extension theorem, partitions, poset structures

\smallskip
{\bf MSC (2010):} 94B05, 94B99, 16L60

%%%%%%%%%%%%%%%%%%%%%%%%%%%%%%%%%%%%%%%%%%%
\section{Introduction}\label{SS-Intro}
%%%%%%%%%%%%%%%%%%%%%%%%%%%%%%%%%%%%%%%%%%%

In her thesis~\cite{MacW62}, MacWilliams showed that every Hamming weight-preserving linear isomorphism~$f:\cC\longrightarrow\cC'$
between codes in~$\F^n$, where~$\F$ is a finite field,
is a monomial map; that is,~$f$ is given by a permutation and a rescaling of the codeword coordinates.
Since this is equivalent to saying that~$f$ can be extended to a Hamming weight-preserving isomorphism on~$\F^n$, this result is referred to as
the MacWilliams extension theorem.

The theorem can be reformulated in the following way.
Any two vectors~$x,y\in\F^n$ have the same Hamming weight if and only if $y=xM$ for some monomial matrix~$M$ over~$\F$.
As a consequence, a linear map $f:\cC\longrightarrow\cC'$ preserves the Hamming weight if and only if for every $x\in\cC$ there exists
a monomial matrix~$M_x$ such that $f(x)=xM_x$.
The MacWilliams extension theorem then says that every such \emph{local} (or \emph{pointwise}) monomial map extends to a global monomial map in the sense
that there exists a global monomial matrix $M$ such that $f(x)=xM$ for all $x\in\cC$.

In this paper we will investigate, for a finite ring~$R$, which classes of left $R$-linear maps between codes in the left $R$-module~$R^n$
satisfy the MacWilliams extension theorem, and which subgroups~$\cU$ of $\glnr$ satisfy the local-global property in the sense that every pointwise
$\cU$-map, defined on some code in~$R^n$, is a global $\cU$-map.
Note that this implies that the map can be extended to an isomorphism on $R^n$ induced by some matrix $M\in\cU$.
Of course, we are interested in subgroups that preserve certain properties just like the monomial matrices preserve the Hamming weight.

The question whether certain classes of maps satisfy the extension property has been vastly studied in the literature.
Beginning with~\cite{Gol80}, Goldberg proved that for any subgroup~$U$ of~$\F^\times$ every local $U$-monomial map is a global $U$-monomial map.
In the same paper the author also posed exactly the questions that we have raised above: which subgroups of $\GL(n,\F)$ satisfy the local-global
property and what are the ``weight functions'' preserved by the associated maps?

In~\cite{Wo97,Wo99} Wood generalized the classical MacWilliams extension theorem as well as Goldberg's result to Frobenius rings.
We will recover these results later in Theorem~\ref{T-MonUR} and Theorem~\ref{T-MacWExtHamm}.
By providing several examples, we will see that, not surprisingly, the Frobenius property of the underlying ring is an indispensable
requirement for the local-global property in all cases under investigation.
For instance, we will see that the general linear group has the local-global property if~$R$ is Frobenius, but may lack this property for
non-Frobenius rings (see also the paragraph after Example~\ref{E-NonFrobExa} for a brief comment on QF rings).
This observation is in line with earlier results by Dinh/L\'{o}pez-Permouth~\cite{DiLP04a} and Wood~\cite{Wo08} which show that the MacWilliams
extension theorem for the Hamming weight holds true if and only if the underlying ring is Frobenius.

Recently, a different approach to the MacWilliams extension theorem has been undertaken by Greferath et al.~\cite{GMZ13}.
They study weight functions on rings~$R$ that are finite products of finite chain rings and where the left and right symmetry group of the weight
are both given by~$R^\times$.
They provide a characterization of those weights, for which every linear map between codes in~$R^n$ that preserves
the additively extended weight, extends to a monomial transformation on~$R^n$.

It is interesting to note that the rank weight, prevalent in random network coding, does not satisfy the MacWilliams
extension theorem.
We will provide a simple example in the next section.

After revisiting the key ingredients for our purposes, namely Frobenius rings, partitions, and their character-theoretic duals, we will establish the
local-global property of various subgroups in Section~\ref{SS-LocGlob}.
This will allow us to discuss the MacWilliams extension theorem for several weight-preserving maps in Section~\ref{SS-Goldberg}.
We will also provide a short proof of the well-known extension theorem for maps preserving the homogeneous weight.

The final section is devoted to a discussion of poset weights as introduced by Niederreiter~\cite{Nie91}, Rosenbloom/Tsfasman~\cite{RoTs97},
and Brualdi et al.~\cite{BGL95}.
After the discovery of the fruitfulness of such weight functions in coding theory, see for instance~\cite{Skr07} and~\cite{BFSF13} and the references therein,
poset structures for codes over fields have gained a lot of attention.
For instance, MacWilliams identity theorems have been established for various settings in Dougherty/Skriganov~\cite{DoSk02}, Kim/Oh~\cite{KiOh05},
Pineiro/Firer~\cite{PiFi12} and, for additive codes, in~\cite{GL13pos}.
The isometry group~of $\F^n$ has been derived by Lee~\cite{Lee03} for the special case of the Niederreiter-Rosenbloom-Tsfasman weight
and by Panek et al.~\cite{PFKH08} for general poset weights.
These results have been utilized by Barg et al.~\cite{BFSF13} in order to study duality of association schemes that arise from such isometry groups.
They also investigate the extension problem for order-preserving bijections between poset ideals.

In this paper we will show that the MacWilliams extension theorem holds true for poset weight-preserving maps over Frobenius rings if and only if the
poset is hierarchical.
In the latter case, this will also lead to the isometry group of the ambient space~$R^n$.
This generalizes the aforementioned result by Panek et al.~\cite{PFKH08} to Frobenius rings in the case where the poset is hierarchical.

%%%%%%%%%%%%%%%%%%%%%%%%%%%%%%%%%%%%%%%%%%%
\section{Motivation and Examples}\label{SS-Motiv}
%%%%%%%%%%%%%%%%%%%%%%%%%%%%%%%%%%%%%%%%%%%
In this section we specify the notions discussed in the introduction, most importantly the local-global property, and present some
examples where this property fails.
Later in Section~\ref{SS-LocGlob} we will be in the position to prove the
local-global property for various situations.

Throughout this section let~$R$ be a finite ring with unity.
Its group of units is denoted by $R^\times$, and the general linear group of order~$n$ over~$R$ is defined as
$\GL(n,R):=\{A\in R^{n\times n}\mid \exists\, B\in R^{n\times n}: AB=I\}
  =\{A\in R^{n\times n}\mid \exists\, B\in R^{n\times n}: AB=BA=I\}$.

We consider~$R^n$ as a left~$R$-module consisting of row vectors.
Submodules of~$R^n$ are meant to be left $R$-submodules and linear maps are left $R$-linear maps.
Occasionally, we will also need~$R^n$ as a right $R$-module, in which case this will always be made clear.

Recall that a \emph{monomial matrix} is a square matrix where each column and row has exactly one non-zero entry
and these non-zero elements are units.

%%%%%%%%%%%%%%%%%%%%%%%%%%%%
\begin{defi}\label{D-MonMat}
For any subgroup~$U$ of the multiplicative group~$R^\times$ denote by $\Mon_U(n,R)$ the group of monomial $(n\times n)$-matrices over~$R$
whose nonzero elements are in~$U$.
We write $\Mon(n,R)$ for the group $\Mon_{R^\times}(n,R)$.
\end{defi}
%%%%%%%%%%%%%%%%%%%%%%%%%%%%

As explained in the introduction, a linear map $f:\cC\longrightarrow \cC'$ is Hamming weight-preserving if and only if for all $x\in\cC$ there
exists a monomial matrix $M_x\in\MonF$ such that $f(x)=xM_x$.
The MacWilliams extension theorem states that there exists a global $M\in\MonF$ such that $f(x)=xM$ for all $x\in\cC$.
This motivates the following definition.

%%%%%%%%%%%%%%%%%%%%%%%%%%
\begin{defi}\label{D-LocGlobMap}
Let $\cC\subseteq R^n$ be a code (i.e., a left $R$-submodule of~$R^n$), and let $\cU\leq\GL(n,R)$ be a subgroup of the general linear group over~$R$.
A linear map $f:\cC\longrightarrow R^n$ is called a \emph{local $\cU$-map} if for every $x\in\cC$ there exists a matrix $M_x\in\cU$ such
that $f(x)=xM_x$.
The map~$f$ is called a \emph{global $\cU$-map} if there exists a universal matrix $M\in\cU$ such that $f(x)=xM$ for all $x\in\cC$.
We say that the group~$\cU$ has the \emph{local-global property} if every local $\cU$-map $f:\cC\longrightarrow R^n$, where $\cC\subseteq R^n$ is any code,
is a global $\cU$-map.
Expressed differently, the local~$\cU$-maps \emph{satisfy the MacWilliams extension theorem}.
\end{defi}
%%%%%%%%%%%%%%%%%%%%%%%%%%%%

By definition every local $\cU$-map is linear and injective, and hence an isomorphism onto its image.
Its inverse is also a local $\cU$-map.
Evidently, a group~$\cU$ has the local-global property if and only if every local $\cU$-map $f:\cC\longrightarrow R^n$ can be
extended to an isomorphism $\hat{f}: R^n\longrightarrow R^n$ induced by some $U\in\cU$.
This justifies the terminology above referring to the MacWilliams extension theorem.

Cast in the above language, MacWilliams' classical extension theorem states that the group $\MonF$ has the local-global property.
A generalization of this result was presented by Goldberg.

%%%%%%%%%%%%%%%%%%%%%%%%%%%%%%%%%
\begin{theo}[\mbox{\cite[p.~367]{Gol80}}]\label{T-Goldberg}
Let~$U$ be a subgroup of the multiplicative group~$\F^\times$.
Then~$\mbox{\rm Mon}_U(n,\F)$ satisfies the local-global property.
\end{theo}
%%%%%%%%%%%%%%%%%%%%%%%%%%%%%%%%%%%%

In~\cite[Thm.~6.3]{Wo99} and \cite[Thm.~10]{Wo97}, Wood generalized the MacWilliams extension theorem and Goldberg's result to
codes over finite Frobenius ring (more on Frobenius rings in the next section).
We will present the proofs in our language in Theorems~\ref{T-MonUR} and~\ref{T-MacWExtHamm}.

The last result motivated Goldberg~\cite{Gol80} to ask which subgroups of $\GL(n,\F)$ satisfy the local-global property and which weight functions
are preserved by these groups.
We formulate the question in the ring setting.

%%%%%%%%%%%%%%%%%%%%%%%%%%%%%%
\begin{question}\label{Q-Goldberg}
Find subgroups $\cU$ of $\GL(n,R)$ together with a ``weight function"~$w:R^n\longrightarrow\C^k$ such that~$w$ is constant on each $\cU$-orbit
(i.e.,~$\cU$ preserves~$w$) and every $w$-preserving linear isomorphism $f:\cC\longrightarrow\cC'$ between codes in~$R^n$
is a global $\cU$-map.
For a function~$w$ satisfying this property, we say that the \emph{$w$-preserving maps satisfy the MacWilliams extension theorem}.
\end{question}
%%%%%%%%%%%%%%%%%%%%%%%%%%%%%%%%%%%%
Mostly, our weight function will take values in~$\C$, but on a few occasions also values in some~$\C^k$ will be considered.

We will first study the following related problem.
In Section~\ref{SS-Goldberg} we will return to the above question.

%%%%%%%%%%%%%%%%%%%%%%%%%%%%%%%%%%
\begin{problem}\label{P-LocGlob}
Which subgroups $\cU$ of $\mathrm{GL}(n,R)$ satisfy the local-global property?
\end{problem}
%%%%%%%%%%%%%%%%%%%%%%%%%%%%%%%%%

In order to illustrate that the local-global property is not trivial, we present a few examples.
The first one shows that even a local $\cU$-map defined on all of $\F^n$, where~$\F$ is a field, need not be a global $\cU$-map.

%%%%%%%%%%%%%%%%%%%%%%%%%%%%%%%%%%%%%%
\begin{exa}\label{E-F3}
Let $\F_3$ be the field with three elements and consider the group
\[
   \cU=\left\{\begin{pmatrix}a&b\\0&c\end{pmatrix}\,\bigg|\, a,b,c\in \F_3\text{ and } ac=1\right\}\leq\GL(2,\F_3).
\]
Let~$f$ be the linear map
\[
 f:\F_3^2\longrightarrow \F_3^2,\quad (\alpha,\beta)\longmapsto (\alpha,\beta)\begin{pmatrix}2&1\\0&1\end{pmatrix}.
\]
Then~$f$ is obviously not a global $\cU$-map.
But~$f$ is a local~$\cU$-map, as follows from the identities
\[
  f(0,1)=(0,1)\begin{pmatrix}1&0\\0&1\end{pmatrix},\
  f(1,0)=(1,0)\begin{pmatrix}2&1\\0&2\end{pmatrix},\
  f(1,1)=(1,1)\begin{pmatrix}2&0\\0&2\end{pmatrix},\
  f(1,2)=(1,2)\begin{pmatrix}2&2\\0&2\end{pmatrix},
\]
along with the fact that each vector in~$\F_3^2$ is a scalar multiple of one of the vectors $(0,1),(1,0)$, $(1,1),(1,2)$.
This example also shows that~$f$ is a local $\mathrm{SL}(2,\F_3)$-map, but not a global one.
Thus, both the groups~$\cU$ and $\mathrm{SL}(2,\F)$ do not satisfy the local-global property.
\end{exa}
%%%%%%%%%%%%%%%%%%%%%%%%%%%%%%%%%%%%%%%%%%%

The next example illustrates that even the general linear group $\GL(n,R)$ may not satisfy the local-global property.
We will see later that the lack of the local-global property is due to the fact that the ring~$R$ in this example is not
Frobenius.

%%%%%%%%%%%%%%%%%%%%%%%%%%%%%%%%%%%%%%%%%%
\begin{exa}\label{E-GLLocNonGlob}
Consider the commutative ring $R=\F_2[x,y]/(x^2,xy,y^2)=\{0,x,y,x+y,1,1+x,1+y,1+x+y\}$.
Writing $\cA:=\{0,x,y,x+y\}$ and $\cB:=\{1,1+x,1+y,1+x+y\}$, we have
$R=\cA\cup\cB$.
Stated differently, $R=\{\alpha+a\mid \alpha\in\F_2,\,a\in\cA\}$.
Note that $u^2=1,\ ab=0$ and $au=a$ for all $u\in\cB$ and $a,b\in\cA$.
Consider the code $\cC\subseteq R^2$ generated by $(x,y)$ and $(y,x+y)$.
Thus, $\cC=\{(0,0),\,(x,y),\,(y,x+y),\,(x+y,x)\}$.
Let $f:\cC\longrightarrow R^2$ be the linear map given by
\begin{align*}
   &f(x,y)=(x,y)\begin{pmatrix}1\!&\!1 \\0\!&\!1\end{pmatrix}=(x,x+y),\\
   &f(y,x+y)=(y,x+y)\begin{pmatrix}0\!&\!1\\1\!&\!0\end{pmatrix}=(x+y,y),\\
   &f(x+y,x)=(x+y,x)\begin{pmatrix}1\!&\!0\\1\!&\!1\end{pmatrix}=(y,x).
\end{align*}
Observe that this map is indeed $R$-linear. By construction,~$f$ is a local $\GL(2,R)$-map.
We want to show that it is not a global $\GL(2,R)$-map.
In order to do so, note first that every matrix in $\GL(2,R)$ is the sum of a matrix with entries in~$\cA$ and a matrix in $\GL(2,\F_2)$.
Using that
\[
   (v,w)\begin{pmatrix}a\!&\!b\!\\c\!&\!d\end{pmatrix}=(0,0)\ \text{ for all $(v,w)\in\cC$ and all $a,b,c,d\in\cA$, }
\]
it remains to show that there is no matrix $A\in\GL(2,\F_2)$ such that $f(v,w)=(v,w)A$ for all $(v,w)\in\cC$.
But this can easily be verified.
All of this proves that~$f$ cannot be extended to an isomorphism on~$R^2$ (and not even to any linear map on~$R^2$), and hence $\GL(2,R)$ does not
have the local-global property.
With the same reasoning, the example also shows that the subgroup $\cU:=\{A\in\GL(2,R)\mid \det(A)=1\}$ does not
have the local-global property either.
\end{exa}
%%%%%%%%%%%%%%%%%%%%%%%%%%%%%%%%

The final example of this series, taken from~\cite[Sec.~2]{Wo97} by Wood, shows that a Hamming weight-preserving map may not even be a \emph{local} $\cU$-map
for any subgroup $\cU\leq\glnr$.
Again, as it has been observed already by Wood, this fails because the ring is not Frobenius; see also Theorem~\ref{T-MacWExtHamm} in this paper.

%%%%%%%%%%%%%%%%%%%%%%%%%
\begin{exa}\label{E-NonFrobExa}
Let~$R$ be the ring of the previous example and let $\cC=\inner{x}=\{0,x\}$ and $\cC'=\inner{y}=\{0,y\}$.
Then~$\cC,\,\cC'$ are codes over~$R$ of length~$1$.  The map $f:\,\cC\longrightarrow\cC'$ defined by $f(0)=0$ and $f(x)=y$ is
an $R$-isomorphism between~$\cC$ and~$\cC'$.
Note that~$f$ is Hamming weight-preserving, but not even a local $R^\times$-map since there is no unit $\alpha\in R^\times$
such that $y=\alpha x$.
This also implies that~$f$ cannot be extended to a monomial map on~$R$.
\end{exa}
%%%%%%%%%%%%%%%%%%%%%%%%%%

It should be noted that the above counterexamples are based on rings that are not even QF (quasi-Frobenius).
In \cite[Thm.~2.3]{Wo08}  Wood shows that the Frobenius property is necessary for the Hamming weight to satisfy the 
extension theorem, and he provides an explicit QF ring with a weight-preserving map that fails the extension property~\cite[Ex.~5.1]{Wo08}.
One can easily check that this example also shows that the group $\Mon(n,R)$ does not satisfy the local-global property 
for that particular QF ring. 
We do not know whether, for instance, $\GL(n,R)$ also fails the local-global property.
As we will see in the next sections, the case~$\Mon(n,R)$ is significantly more difficult than other subgroups, and
therefore we leave it open to future research whether our results on local-global properties can be (partially)
extended to QF-rings.
In that case, the methods certainly need to be adapted because we make substantial use of a generating character.

\medskip
Before we go on and set up the general theory for our purposes, we take the opportunity to present yet another instance of weight-preserving maps for
which the MacWilliams extension theorem fails: rank-preserving maps.
They are motivated by the rank weight which is the crucial player for error-correction in
random network coding.
The first of the following two examples deals with matrix codes in $\F^{m\times n}$ endowed with the obvious rank function,
whereas the second example deals with codes in $(\F_{q^m})^n$ endowed with the rank over~$\F_q$.
As to our knowledge, such examples, though very simple, have not been presented yet in the literature.
For the isometry groups of the ambient spaces~$\F^{m\times n}$ and~$(\F_{q^m})^n$ and for further background on this topic, we refer to
Berger~\cite{Ber03} and Morrison~\cite{Mor13} and the vast literature on random network coding.
%%%%%%%%%%%%%%%%%%%%%%%%%%%%%%%
\begin{exa}\label{E-RankMetric}
\begin{alphalist}
\item Let $\F$ be a finite field and consider the vector space $\F^{m\times n}$ endowed with the rank metric, i.e.,
      $\wtRK(A):=\rank_{\F}(A)$.
      Then the $\wtRK$-preserving linear maps do not satisfy the MacWilliams extension theorem.
      Take for instance $\F=\F_2$, and consider the code
      \[
         \cC=\big\{(A\,|\,0_{2\times1})\,\big|\, A\in\F^{2\times2}\big\}\subseteq\F^{2\times3}.
      \]
      The linear map $f:=\cC\longrightarrow \F^{2\times3},\ (A\,|\,0_{2\times1})\longmapsto (A\T\,|\,0_{2\times1})$, is clearly
      $\wtRK$-preserving, yet it cannot be extended to a $\wtRK$-preserving isomorphism on~$\F^{2\times3}$.
      The latter can easily be seen using the fact
      that $\Smallsixmat{0}{0}{1}{0}{0}{0}$ has to be mapped to a matrix with rank~$1$, which along with linearity leads to a contradiction to
      \[
        f\begin{pmatrix}1&0&0\\0&0&0\end{pmatrix}=\begin{pmatrix}1&0&0\\0&0&0\end{pmatrix}\text{ and }
        f\begin{pmatrix}0&1&0\\0&0&0\end{pmatrix}=\begin{pmatrix}0&0&0\\1&0&0\end{pmatrix}.
      \]
      \\
      Of course, the entire situation can also be cast in the setting of Definition~\ref{D-LocGlobMap}.
      Note first that $A,B\in\F^{m\times n}$ satisfy the equivalence $\wtRK(A)=\wtRK(B)\Longleftrightarrow B=UAV$ for some $U\in\GL(m,\F)$ and $V\in\GL(n,\F)$.
      Using the vectorization $\Vec(M)\in\F^{mn}$, which is the vector listing the rows of~$M\in\F^{m\times n}$ in the given order,
      we obtain the isometric space~$(\F^{mn},\wtRK)$.
      With the aid of the Kronecker product $\otimes$ for matrices, the above equivalence now reads as
      $\wtRK(\Vec(A))=\wtRK(\Vec(B))\Longleftrightarrow \Vec(B)=\Vec(A)(U\T\otimes V)$ for some $U\in\GL(m,\F)$ and $V\in\GL(n,\F)$.
      Thus a $\wtRK$-preserving map is simply a local $G$-map, where $G:=\GL(m,\F)\otimes\GL(n,\F)$, which is a subgroup of
      $\GL(mn,\F)$.
      The above example shows that~$G$ does not have the local-global property.
\item Similar to~(a), consider the code
      $\cC=\big\{\text{diag}(A,B)\mid A,\,B\in\F^{2\times2}\big\}\subseteq\F^{4\times4}$ and
      the $\wtRK$-preserving map $f:=\cC\longrightarrow\F^{4\times4},\ \text{diag}(A,B)\longmapsto\text{diag}(A\T,B)$.
      Again, it can be shown that~$f$ cannot be extended to a $\wtRK$-preserving isomorphism on~$\F^{4\times4}$.
      This must not be confused with a recent result by Greferath et al.~\cite{GHMWZ13}
      which states that the rank weight on matrix rings $R:=\F^{m\times m}$ satisfies the extension theorem;
      see \cite[Ex.~4.7]{GHMWZ13}.
      The latter applies -- as always, when formulated in this way -- to $R$-linear maps between modules over~$R$.
      The above code~$\cC$ and the map~$f$ are merely $\F$-linear.
\item Another setting  appearing in random network coding is as follows.
      For $x=(x_1,\ldots,x_n)\in(\F_{q^m})^n$ define $\wtRKq(x):=\dim_{\F_q}\langle x_1,\ldots,x_n\rangle$,
      where the latter denotes the $\F_{q}$-subspace generated by $x_1,\ldots,x_n$ in $\F_{q^m}$.
      The $\F_{q^m}$-linear and $\wtRKq$-preserving maps do not satisfy a Mac\-Wil\-liams extension theorem:
      Consider the field $\F_{2^4}$ with primitive element $\omega$ satisfying $\omega^4+\omega+1=0$, and let $\cC=\{a(1,\omega)\mid a\in\F_{2^4}\}$ and
      $\cC'=\{a(1,\omega^5)\mid a\in\F_{2^4}\}$. Then one can verify that the $\F_{2^4}$-linear map
      $f:\cC\longrightarrow\cC',\ a(1,\omega)\longmapsto a(1,\omega^5)$ is $\wtRKq$-preserving, but cannot be extended to an $\F_{2^4}$-linear
      $\wtRKq$-isometry on $(\F_{2^4})^2$.
      This can be seen by testing all vectors with rank~$1$ as the potential image of $(0,1)\in(\F_{2^4})^2$.
      In the same way one can verify that for each~$\sigma$ in the Galois group $\text{Gal}(\F_{2^4}|\F_2)$ the semi-linear map
      $\cC\longrightarrow\cC',\ a(1,\omega)\longmapsto \sigma(a)(1,\omega^5)$ is $\wtRKq$-preserving, but cannot be extended to a semi-linear
      $\wtRKq$-isometry on $(\F_{2^4})^2$.
\end{alphalist}
\end{exa}
%%%%%%%%%%%%%%%%%%%%%%%%%%%%%%

%%%%%%%%%%%%%%%%%%%%%%%%%%%%%%%%%%%%%%%%%%%
\section{Basic Notions on Characters and Frobenius Rings}\label{SS-Basics}
%%%%%%%%%%%%%%%%%%%%%%%%%%%%%%%%%%%%%%%%%%%
In this brief section we recall the basic notions and properties of characters of finite abelian groups and of Frobenius rings.

Let~$A$ be a finite abelian group.
A \emph{character}~$\chi$ on $A$ is a group homomorphism $\chi:A\longrightarrow \C^*$, where $\C^*$ is the multiplicative group
of nonzero complex numbers.
The set of all characters has a group structure with addition $(\chi_1\oplus\chi_2)(a):=\chi_1(a)\chi_2(a)$.
We denote this character group by $\widehat{A}$.
Its zero element is the trivial map $\chi\equiv1$.
It is well-known that the groups~$A$ and~$\widehat{A}$ are (non-canonically) isomorphic.
Thus $|A|=|\widehat{A}|$.
On the other hand, $\widehat{\phantom{\big|}\hspace*{.5em}}\hspace*{-.9em}\widehat{A}$
can and will be identified with~$A$ via the map $a(\chi):=\chi(a)$.

Linear independence of characters in the following form will play a crucial role later on.

%%%%%%%%%%%%%%%%%%
\begin{prop}\label{P-LinIndepChar}
Let $\chi_1,\ldots,\chi_N$ and $\psi_1,\ldots,\psi_M$ be (not necessarily distinct) characters of~$A$.
Assume that
\[
     \sum_{i=1}^N\chi_i=\sum_{i=1}^M\psi_i
\]
in the vector space $\C^A$ of complex-valued maps on~$A$.
Then the multisets $\{\!\{\chi_1,\ldots,\chi_N\}\!\}$ and $\{\!\{\psi_1,\ldots,\psi_M\}\!\}$ coincide.
\end{prop}
%%%%%%%%%%%%%%%%%%%
\begin{proof}
Let $\phi_1,\ldots,\phi_L$ be all the distinct characters among $\chi_1,\ldots,\chi_N,\psi_1,\ldots,\psi_M$. Then the above identity may be written as
\[
    \sum_{i=1}^L\alpha_i\phi_i=\sum_{i=1}^L\beta_i\phi_i
\]
for suitable coefficients $\alpha_i,\,\beta_i\in\Z$.
Now the linear independence of distinct characters in the vector space~$\C^A$, see \cite[p.~291]{Ja85}, implies $\alpha_i=\beta_i$ for all $i=1,\ldots,L$, and
this proves the desired statement.
\end{proof}

Now we turn to finite rings~$R$ with unity.
Let~$\widehat{R^n}$ be the character group of the group $(R^n,+)$.
Then~$\widehat{R^n}$ can be turned into a left and right $R$-module as follows.
For $r\in R$ and $\chi\in\widehat{R^n}$ define the characters $r\chi$ and $\chi r$ in $\widehat{R^n}$ by
\begin{equation}\label{e-rchir}
    (r\chi)(x)=\chi(xr) \text{ and }(\chi r)(x)=\chi(rx)  \text{ for all }x\in R^n.
\end{equation}
Along with the addition~$\oplus$ of the character group $\widehat{R^n}$, this induces a left and a right $R$-module structure
on~$\widehat{R^n}$.

From Honold~\cite[Thm.~1 and~2]{Hon01} (see also Lamprecht~\cite{Lamp53}, Hirano~\cite{Hi97}, and Wood~\cite{Wo99})
we know that a finite ring~$R$ is \emph{Frobenius} if
$R/\text{rad}(R)\cong\text{soc}(\!_{R}R)$ as left $R$-modules, where $\text{rad}(R)$ is the Jacobson radical of~$R$ and
$\text{soc}(\!_{R}R)$ is the left socle of~$R$.
This property is equivalent to the analogous right-sided one~\cite{Hon01}.
It also follows from~\cite[p.~409]{Hon01}
that~$R$ is Frobenius if and only if~$\widehat{R}$ and~$R$ are isomorphic left $R$-modules.
In other words,~$R$ is Frobenius if and only if there exists a character~$\chi\in\widehat{R}$ such that
\begin{equation}\label{e-RRhat}
  R\longrightarrow\widehat{R},\ r\longmapsto r\chi,
\end{equation}
is an isomorphism of left $R$-modules (and thus~$\widehat{R}$ is a free left $R$-module with basis~$\chi$).
Such a character is called a \emph{left generating character}.
It is known~\cite[Thm.~4.3]{Wo99} that a character is left generating if and only if it is right generating.
Hence the left isomorphism in~\eqref{e-RRhat} also implies the according isomorphism
$R\longrightarrow\widehat{R},\ r\longmapsto \chi r$ of right $R$-modules.
We will call such a character~$\chi$ simply a \emph{generating character of}~$R$.

The following proposition will be useful.
%%%%%%%%%%%%%%%%%%%%%%%%%%%%%%%%%%%%%%%%%%%%%%%
\begin{prop}[\mbox{\cite[Cor.~3.6]{ClGo92} and \cite[Prop.~4.2]{Wo99}}] \label{P-GenChar}
Let~$\chi$ be a generating character of~$R$.
\begin{arabiclist}
\item If $I$ is a left or right ideal of $R$ and $I\subset \ker \chi:=\{r\in R\mid \chi(r)=1\}$, then $I=0$.
\item Let~$V$ be a left $R$-module.
         Denote by ${\rm Hom}_R(V,R)$ the group of left $R$-linear maps from~$V$ to~$R$
         and by~$\widehat{V}$ the character group of $(V,+)$.
         The map ${\rm Hom}_R(V,R)\longrightarrow\widehat{V},\;
        g\longmapsto \chi\circ g$ is an injective group homomorphism.
\end{arabiclist}
\end{prop}
%%%%%%%%%%%%%%%%%%%%%%%%%%%%%%%%%%%%%%

Examples of  finite Frobenius rings are finite fields, integer residue ring $\mathbb{Z}_N$, finite chain rings, and
matrix rings $R^{n\times n}$ as well as finite group rings $R[G]$ over a Frobenius ring~$R$.
Direct products of finite Frobenius rings are Frobenius.
The ring in Example~\ref{E-GLLocNonGlob} is not Frobenius;  see~\cite[Ex.~3.2]{ClGo92}.

On the $R$-bimodule~$R^n$ denote by $\inner{\,\cdot\,,\,\cdot\,}$ the standard inner (dot) product.
Using~\eqref{e-rchir} it is easy to see that if~$R$ is a Frobenius ring with generating character~$\chi$,
then~\eqref{e-RRhat} extends to the left $R$-module isomorphism
\begin{equation}\label{e-RRhatnl}
    \alpha_{l}:\; R^n\cong\widehat{R^n},\quad x\longmapsto \chi(\inner{-,x}).
\end{equation}
Similarly, we have the right $R$-module isomorphism
\begin{equation}\label{e-RRhatnr}
     \alpha_{r}:\; R^n\cong\widehat{R^n},\quad x\longmapsto \chi(\inner{x,-}).
\end{equation}

We close this section with the following double annihilator property, which will be useful on several occasions.
%%%%%%%%%%%%%%%%%%%%%
\begin{rem}[\mbox{\cite[Theorem 15.1]{Lam99}}]\label{R-FrobAnn}
Let~$R$ be a Frobenius ring.
For a right (resp.\ left) ideal~$I$ of~$R$ define the left annihilator as
$\text{ann}_{l}(I):=\{r\in R\mid ra=0\text{ for all }a\in I\}$ (resp.\ the
right annihilator as $\text{ann}_{r}(I):=\{r\in R\mid ar=0\text{ for all }a\in I\}$).
Then $\text{ann}_{r}(\text{ann}_{l}(I))=I$ for each right ideal~$I$ and
$\text{ann}_{l}(\text{ann}_{r}(I))=I$ for each left ideal~$I$ of~$R$.
\end{rem}
%%%%%%%%%%%%%%%%%%%%%

%%%%%%%%%%%%%%%%%%%%%%%%%%%%%%%%%%%%%%%%%%
\section{Partitions and Their Dual Partitions}
%%%%%%%%%%%%%%%%%%%%%%%%%%%%%%%%%%%%%%%%
We introduce character-theoretic dualization of partitions and discuss some crucial properties.
Throughout this section let~$R$ be a finite Frobenius ring.

Let us fix some basic terminology.
The sets of a partition $\cP=(P_m)_{m=1}^M$ are called its \emph{blocks}, and we write $|\cP|$ for the number of
blocks in~$\cP$.
Recall that two partitions~$\cP$ and~$\cQ$ are called \emph{identical} if $|\cP|=|\cQ|$ and the blocks coincide
after suitable indexing.
Moreover,~$\cP$ is called \emph{finer} than~$\cQ$, written as $\cP\leq\cQ$, if
every block of~$\cP$ is contained in a block of~$\cQ$.
Note that if~$\cP\leq\cQ$ then $|\cP|\geq|\cQ|$.
Denote by~$\widesim_{\cP}$ the equivalence relation induced by~$\cP$, thus,
$v\widesim_{\cP}v'$ if $v,\,v'$ are in the same block of~$\cP$.

The following notion of partition duality for abelian groups has been proven to be at the core of MacWilliams identities.
It has been introduced for Frobenius rings (in a left-sided variant) by Byrne et~al.~\cite[p.~291]{BGO07} and goes back to
the notion of F-partitions as introduced by Zinoviev/Ericson in~\cite{ZiEr96}; see also~\cite{ZiEr09}.
Reflexive partitions, as defined below, are exactly the partitions that induce abelian association schemes as
studied in a more general context by Delsarte~\cite{Del73}, Camion~\cite{Cam98}, and others.
For an overview of these various approaches and their relations in the language of partitions see also~\cite{GL13pos}.

%%%%%%%%%%%%%%%%%%%%%%%%%%%%%%%%%
\begin{defi}\label{D-DualPartGroup}
Let~$A$ be a finite abelian group and  $\cP=(P_m)_{m=1}^M$ be a partition of~$A$.
The \emph{dual partition} of~$\cP$, denoted by~$\wcP$, is the partition of~$\widehat{A}$ defined via the equivalence relation
\begin{equation}\label{e-simPhat}
  \chi\widesim_{\wcP} \chi' :\Longleftrightarrow \sum_{a\in P_m}\chi(a)=\sum_{a\in P_m}\chi'(a)
  \text{ for all }m=1,\ldots,M.
\end{equation}
The partition~$\cP$ is called \emph{reflexive} if $\wwcP=\cP$.
\end{defi}
%%%%%%%%%%%%%%%%%%%%%%%%%%%%%%%%%%%%%%

Note that $\wwcP$ is a partition of~$A$ due to $A=\widehat{\phantom{\big|}\hspace*{.5em}}\hspace*{-.9em}\widehat{A}$.
The complex numbers $\sum_{a\in P_m}\chi(a)$ are known as the Krawtchouk coefficients of the pair $(\cP,\wcP)$.
They occur in the MacWilliams identities for partition enumerators of codes and their duals, see for instance~\cite{GL13pos}
and the references therein.

The following criterion for reflexivity turns out to be very convenient.
%%%%%%%%%%%%%%%%%%%%%%%%
\begin{theo}[\mbox{\cite[Thm.~3.1]{GL13pos}}]\label{T-ReflCrit}
For any partition~$\cP$ on $A$
we have $|\cP|\leq|\wcP|$ and
$\widehat{\phantom{\big|}\hspace*{.6em}}\hspace*{-.9em}\wcP\leq\cP$.
Moreover, $\cP$ is reflexive if and only if $|\cP|=|\wcP|$.
\end{theo}
%%%%%%%%%%%%%%%%%%%%%%%%%%%%

We will need duality mainly for partitions of $R^n$.
Since~$R$ is a Frobenius ring, we may identify the character-dual~$\widehat{R^n}$ with~$R^n$ as in~\eqref{e-RRhatnl} or as
in~\eqref{e-RRhatnr}.
Either way allows us to define the dual of a partition in~$R^n$ as a partition in~$R^n$.
In order to obtain a useful duality theory, we will need both identifications.
The following definition is taken from~\cite{GL13pos} and adapted to the non-commutative case.

%%%%%%%%%%%%%%%%%%%%%%%%%%%%%%
\begin{defi}\label{D-DualPart}
Fix a generating character~$\chi$ of~$R$.
Let $\cP=(P_m)_{m=1}^M$ be a partition of $R^n$.
The \emph{left $\chi$-dual partition} of~$\cP$ and the \emph{right $\chi$-dual partition} of~$\cP$, denoted
by~$\wcPchil$ and~$\wcPchir$, are defined as~$\alpha_l^{-1}(\wcP)$ and $\alpha_r^{-1}(\wcP)$, respectively.
In other words, $\wcPchil$ is the partition of~$R^n$ given by the equivalence relation
\begin{equation}\label{e-simPhat2l}
  v\widesim_{\wcP^{^{\scriptscriptstyle[\chi,l]}}} v' :\Longleftrightarrow
  \sum_{w\in P_m}\chi(\inner{w,v})=\sum_{w\in P_m}\chi(\inner{w,v'}) \text{ for all }m=1,\ldots,M,
\end{equation}
while~$\wcPchir$ is given by
\begin{equation}\label{e-simPhat2r}
  v\widesim_{\wcP^{^{\scriptscriptstyle[\chi,r]}}} v' :\Longleftrightarrow
  \sum_{w\in P_m}\chi(\inner{v,w})=\sum_{w\in P_m}\chi(\inner{v',w}) \text{ for all }m=1,\ldots,M.
\end{equation}
\end{defi}
%%%%%%%%%%%%%%%%%%%%%%%%%%%%%

It is not hard to find examples showing that the left- and right-dual of a given partition do not coincide in general.
Furthermore, the dual partitions depend on the choice of the generating character~$\chi$.
This is even the case for commutative rings;  for an example see~\cite[Ex.~5.7]{GL13pos}.

We have the following relation between the left- and right-dual partitions.
For reflexive partitions this may be regarded as an analogue of the double annihilator property described in Remark~\ref{R-FrobAnn}.

%%%%%%%%%%%%%%%%%%%
\begin{prop}\label{P-leftrightbidual}
Let~$\cP$ be a partition on~$R^n$.
Then
\begin{equation}\label{e-bidual}
 \widehat{\phantom{\Big|}\hspace*{1.4em}}^{\,\scriptscriptstyle[\chi,r]}\hspace*{-3.5em}\widehat{\cP}^{^{\scriptscriptstyle[\chi,l]}}
 \hspace*{1.2em}
 =\wwcP
 =\;\widehat{\phantom{\Big|}\hspace*{1.4em}}^{\,\scriptscriptstyle[\chi,l]}\hspace*{-3.3em}\widehat{\cP}^{^{\scriptscriptstyle[\chi,r]}}\hspace*{1.2em},
\end{equation}
where~$\wwcP$ is the bidual partition in the group sense of~\eqref{e-simPhat}.
Consequently,~$\cP$ is reflexive if and only if
$\cP=\;\widehat{\phantom{\Big|}\hspace*{1.4em}}^{\,\scriptscriptstyle[\chi,r]}\hspace*{-3.5em}\widehat{\cP}^{^{\scriptscriptstyle[\chi,l]}}\hspace*{1.2em}$, which is equivalent to
$\cP=\;\widehat{\phantom{\Big|}\hspace*{1.4em}}^{\,\scriptscriptstyle[\chi,l]}\hspace*{-3.5em}\widehat{\cP}^{^{\scriptscriptstyle[\chi,r]}}\hspace*{1.2em}$.
Moreover, $\cP=\wcPchil$ if and only if $\cP=\wcPchir$.
We call~$\cP$ $\chi$-self-dual if $\cP=\wcPchil$.
\end{prop}
%%%%%%%%%%%%%%%%%%
\begin{proof}
Let $\cQ=\big(Q_m\big)_{m=1}^M =\wcPchil$ and let $\cR=\widehat{\cQ}^{^{[\chi,r]}}$.
Then $v,\,v'\in R^n$ satisfy $v\widesim[2]_{\cR}v'$ if and only if $\alpha_r(v)\widesim_{\widehat{\cQ}}\alpha_r(v')$, which means
$\sum_{w\in Q_m}\alpha_r(v)(w)=\sum_{w\in Q_m}\alpha_r(v')(w)$ for all $m=1,\ldots,M$.
We compute
\[
  \sum_{w\in Q_m}\alpha_r(v)(w)=\sum_{w\in Q_m}\chi(\inner{v,w})=\sum_{w\in Q_m}\alpha_l(w)(v)
  =\sum_{\psi\in\alpha_l(Q_m)}\psi(v).
\]
Since $\psi(v)=v(\psi)$, due to the canonical identification of the abelian group $(R^n,+)$ with its bidual character group,
the above shows that $v\widesim[2]_{\cR}v'\Longleftrightarrow v\widesim[2]_{\widehat{\alpha_l(\cQ)}}v'$.
But $\widehat{\alpha_l(\cQ)}=\wwcP$, and we arrive at
$v\widesim[2]_{\widehat{\phantom{\raisebox{.5ex}{|}}\!\!}\hspace*{-.75em}\widehat{\mathcal P}\;}v'$.
In the same way we obtain the second identity of~\eqref{e-bidual}.

As for the last equivalence, assume $\cP=\wcPchil$.
Then $|\cP|=|\wcP|$ and Theorem~\ref{T-ReflCrit} implies $\cP=\wwcP$.
Now $\cP=\wcPchir$ follows from~\eqref{e-bidual}.
\end{proof}

%%%%%%%%%%%%%%%%%%%
\begin{exa}\label{E-HamPart}
Let~$\cP$ be the Hamming partition on~$R^n$, thus
$\mathcal{P}=(P_i)_{i=0}^n$, where $P_i$ is the set of all vectors in $R^n$ with Hamming weight $i$.
It is well known that $\wcPchil=\wcPchir=\cP$ for all generating characters~$\chi$.
This can be found in many textbooks; for a brief summary in the terminology of partitions and for the Hamming weight
on a group $A_1\times\ldots\times A_n$, see also \cite[Ex.~2.3(c)]{GL13pos}.
One may note that if $R=\F$ is a field, the Hamming partition coincides with the orbits of the right-multiplication action of
$\MonF$ on $\F^n$. This is not true if~$R$ is not a field.
\end{exa}
%%%%%%%%%%%%%%%%%%%%

For further examples and properties of partitions and their duals in the commutative case we refer to~\cite{GL13pos}.

\medskip

Our main interest are partitions that are given by the orbits of a group action on~$R^n$.
In~\cite{GL13homog} it has been shown that these partitions are reflexive for commutative Frobenius rings.
Below we show that the same is true for noncommutative Frobenius rings.
A similar version of that result can also be found in~\cite[Lem.~4.63]{Cam98} by Camion, where it is
derived in the language of association schemes and presented for automorphism groups of abelian groups.
For $n=1$ and $R=\Z_N$ the statement has also been shown by Ericson et al.~\cite[Thm.~1]{ESTZ97}.

%%%%%%%%%%%%%%%%%%%%%%%%%%%%%%
\begin{prop}[See also \mbox{\cite[Prop.~2.11]{GL13homog}}]\label{P-Uorbits}
Let~$\cU$ be a subgroup of $\GL(n,R)$.
Consider the (right and left) group actions
\[
  \rho_r: R^n\times\cU\longrightarrow R^n,\ (x,U)\longmapsto xU \ \text{ and }\
  \rho_l: \cU\times R^n\longrightarrow R^n,\ (U,x)\longmapsto (Ux\T)\T.
\]
Denote by~$\cP_\cU$  and $\cP_{\cU^{\sf T}}$ the partitions of~$R^n$
given by the orbits of the actions~$\rho_r$ and~$\rho_l$, respectively.
Then $\cP_{\cU^{\sf T}}=\widehat{\cP_{\cU}}^{\scriptscriptstyle[\chi,l]}$ and
$\cP_{\cU}=\widehat{\cP_{\cU^{\sf T}}}^{\scriptscriptstyle[\chi,r]}$
for each generating character~$\chi$ of~$R$.
As a consequence,~$\cP_{\cU}$ and $\cP_{\cU^{\sf T}}$ are reflexive.
\end{prop}
%%%%%%%%%%%%%%%%%%%%%%%%%%%%%
Note that if~$R$ is commutative, then $(Ux\T)\T=xU\T$, and~$\rho_l$ is simply the right action induced by the
transposed group $\cU\T:=\{U\T\mid U\in\cU\}$.

\begin{proof}
Let $v,\,v'\in R^n$ be in the same partition set of~$\cP_{\cU^{\sf T}}$, thus $v'=(Uv\T)\T$ for some $U\in\cU$, and thus
$\inner{w,v'}=wUv\T=\inner{wU,v}$.
Using that~$PU=P$ for each orbit~$P$ of~$\cP_{\cU}$, we obtain
$\sum_{w\in P}\chi(\inner{w,v'})=\sum_{w\in P}\chi(\inner{w,v})$.
This shows that~$\cP_{\cU^{\sf T}}$ is finer than or equal to
$\widehat{\cP_{\cU}}^{\scriptscriptstyle[\chi,l]}$
and thus $|\cP_{\cU^{\sf T}}|\geq|\widehat{\cP_{\cU}}^{\scriptscriptstyle[\chi,l]}|$.
On the other hand, if $v\widesim[2]_{\cP_{\cU}}v'$, then $v'=vU$ for some $U\in\cU$, and therefore
$\inner{v',w}=\inner{v,(Uw\T)\T}$.
This yields $\sum_{w\in Q}\chi(\inner{v',w})=\sum_{w\in Q}\chi(\inner{v,w})$ for each block $Q$ of $\cP_{\cU^{\sf T}}$,
and thus $\cP_{\cU}$ is finer than or equal to $\widehat{\cP_{\cU^{\sf T}}}^{\scriptscriptstyle[\chi,r]}$.
Hence $|\cP_{\cU}|\geq|\widehat{\cP_{\cU^{\sf T}}}^{\scriptscriptstyle[\chi,r]}|$.
With the aid of Theorem~\ref{T-ReflCrit} we conclude
$|\cP_{\cU^{\sf T}}|\geq|\widehat{\cP_{\cU}}^{\scriptscriptstyle[\chi,l]}|\geq|\cP_{\cU}|
  \geq |\widehat{\cP_{\cU^{\sf T}}}^{\scriptscriptstyle[\chi,r]}|\geq|\cP_{\cU^{\sf T}}|$.
Thus, we have equality at each step,  and this results in the desired identities.
Reflexivity follows with the aid of Proposition~\ref{P-leftrightbidual}.
\end{proof}

The above result has the remarkable consequence that the group
actions~$\rho_r$ and~$\rho_l$ lead to the same number of orbits in~$R^n$.
This property is not true in general if~$R$ is not a Frobenius ring.
For instance, for~$R$ as in Example~\ref{E-GLLocNonGlob}, the action on~$R^2$ of the group
\[
   \cU=\bigg\{\begin{pmatrix}1&r\\0&u\end{pmatrix}\,\bigg|\, r\in R,\,u\in R^*\bigg\}\leq \GL_2(R)
\]
leads to~$17$ orbits, whereas $\cU\T$ produces~$20$ orbits.

%%%%%%%%%%%%%%%%%%%%%%%%%%%%%%%%%%%%%%%%%%%%
\section{Subgroups with the Local-Global Property}\label{SS-LocGlob}
%%%%%%%%%%%%%%%%%%%%%%%%%%%%%%%%%%%%%%%%%%%%
In this section we will provide some answers to Problem~\ref{P-LocGlob} by establishing the local-global property for
various subgroups of $\glnr$.
Throughout, let~$R$ be a finite Frobenius ring and let~$\chi$ be a generating character of~$R$.

We start with the following two lemmas that will be needed several times in the future.
The first one is due to Wood~\cite{Wo99} and comes as a consequence of a result of Bass~\cite{Bass64}.

%%%%%%%%%%%%%%%%%%%%%%%%
\begin{lemma}[\mbox{\cite[Prop.~5.1]{Wo99}}]\label{L-CyclicMod}
Let $\mathcal{M}$  be a right module over a finite ring~$S$.
Suppose $x,y\in \mathcal{M}$ generate the same cyclic right $S$-module, i.e., $xS=yS$.
Then $x= y\alpha$ for some unit $\alpha\in S$.
\end{lemma}
%%%%%%%%%%%%%%%%%%%%%%%%%

The following technical lemma will be a crucial step in establishing the local-global property for various subgroups.

%%%%%%%%%%%%%%%%%%%%%%%%%
\begin{lemma}\label{L-technical}
Let $\cU$ be a subgroup of $\glnr$ with associated orbit partition $\cP=\cP_{\cU}$ of $R^n$ (see Proposition~\ref{P-Uorbits}), and
let~$Q$ be any block of the left-dual partition $\wcPchil=\cP_{\cU^{\sf T}}$.
Let $f:\cC\longrightarrow R^n$ be a local $\cU$-map defined on the code $\cC\subseteq R^n$.
Then
\begin{equation}\label{e-sumQ}
\sum_{y\in Q}\chi(\inner{f(x),y})=\sum_{y\in Q}\chi(\inner{x,y})\text{ for all }x\in\cC.
\end{equation}
Moreover, for every $z\in R^n$ there exists a matrix $A_z\in\cU$ such that
\[
    \inner{f(x),z}=xA_z z\T  \text{ for all }x\in\cC.
\]
\end{lemma}
%%%%%%%%%%%%%%%%%%%%%%

\begin{proof}
Since $f$ is a local $\cU$-map, we have $x\sim_{\cP}f(x)$ for every $x\in \cC$.
From Proposition~\ref{P-Uorbits} we know that~$\cP=\widehat{\cP_{\cU^{\sf T}}}^{^{\scriptscriptstyle[\chi,r]}}$.
Now the identity in~\eqref{e-sumQ} follows from~\eqref{e-simPhat2r}.

For the second statement note first that for any fixed $y\in R^n$, the assignments $x\longmapsto \inner{x,y}$ and
$x\longmapsto \inner{f(x),y}$ are (left-)linear maps from $\cC$ to $R$.
As a consequence, $\chi(\inner{-,y})$ and $\chi(\inner{f(-),y})$ are characters of (the additive group of) $\cC$.
Thus, both sides of the identity~\eqref{e-sumQ} consist of sums of character values, and Proposition~\ref{P-LinIndepChar}
tells us that every character that appears on the left hand side
of~\eqref{e-sumQ} must appear on the right hand side as well (and vice versa).
Let now $z\in R^n$.
Then~$z$ is contained in some partition set~$Q$ of $\wcPchil$, and by the above there exists some $v\in Q$ such that
$\chi(\inner{f(-),z})=\chi(\inner{-,v})$.
Since~$\wcPchil$ is the orbit partition of the action~$\rho_l$ due to Proposition~\ref{P-Uorbits},  $v=(A_z z\T)\T$ for some
$A_z\in\cU$.
Hence $\chi(\inner{f(-),z})=\chi(\inner{-,(A_z z\T)\T})$.
Now Proposition \ref{P-GenChar}(2) yields $\inner{f(-),z}=\inner{-,(A_z z\T)\T}$,
as maps from $\cC$ to $R$.
Hence $\inner{f(x),z}=\inner{x,(A_z z\T)\T}=xA_z z\T$ for all $x\in\cC$.
\end{proof}

Now we are ready to investigate various subgroups of $\glnr$  with respect to the local-global property.
We begin with the group $\ltnr$ of all invertible lower triangular matrices.
Its relation to the Rosenbloom-Tsfasman metric will be discussed in the next section.

%%%%%%%%%%%%%%%%%%
\begin{theo}\label{T-ltnr}
The group $\LT(n,R)$ satisfies the local-global property.
\end{theo}
%%%%%%%%%%%%%%%%%%%

\begin{proof}
Let $\cC\subseteq R^n$ be a code and $f:\cC\longrightarrow R^n$ be a local $\mathrm{LT}(n,R)$-map.
Denote by $f_i:\cC\longrightarrow R$ the $i$-th coordinate function of~$f$.
By Lemma~\ref{L-technical} there exists for each standard basis vector $e_i\in R^n$ a matrix $A_i\in \ltnr$
such that $\inner{f(x),e_i}=xA_i e_i\T$ for all $x\in\cC$.
But this means $f_i(x)=xA_i e_i\T$ for all $x\in\cC$.
Define the matrix $B\in R^{n\times n}$ as $B=(A_1 e_1\T,\ldots,A_n e_n\T)$,
i.e., the $i$-th column of~$B$ is given by the $i$-th column of $A_i$.
Since $A_i\in\ltnr$ for all~$i$, the matrix~$B$ is also lower triangular with units on the diagonal.
In other words, $B\in\ltnr$.

By construction, $f(x)=(f_1(x),\ldots,f_n(x))=xB$ for all $x\in\cC$.
This proves that $f$ is a global $\ltnr$-map.
\end{proof}

In the same way one can show that for any subgroup $U\leq R^\times$, the group $\mathrm{LT}_U(n,R)$ consisting of the lower
triangular matrices with elements from~$U$ on the diagonal has the local-global property.

\medskip

The above proof may serve as a blueprint for establishing the local-global property for further groups.
The only step in the previous proof that made explicit use of the lower triangular form is the argument that the
constructed matrix~$B$ is also invertible and lower triangular.
While this was evident in the above case, there is a more general and elegant argument that will make the above construction work
more generally.

We first present the following result, which takes care of the invertibility of the global matrix.

%%%%%%%%%%%%%%%%%%
\begin{lemma}\label{L-Sunits}
Let~$S$ be any subring of $R^{n\times n}$, and let~$\cU(S)$ be the group of invertible matrices in~$S$.
Suppose $f:\cC\longrightarrow\cC'$ is an isomorphism with the property that there exist matrices $A,B\in S$ such that
$f(x)=xA$ for all $x\in\cC$ and $f^{-1}(y)=yB$ for all $y\in\cC'$.
Then $f$ is a global $\cU(S)$-map.
\end{lemma}
%%%%%%%%%%%%%%%%
\begin{proof}
Let~$G\in R^{k\times n}$ be any matrix whose rows generate the code~$\cC$, thus $\cC=\{uG\mid  u\in R^k\}$.
Define $G':=GA\in R^{k\times n}$.
Then $\cC'=\{uG'\mid u\in R^k\}$ and $G=G'B$ due to the assumption on~$f^{-1}$.
This implies that the two right $S$-modules $GS:=\{GM\mid M\in S\}$ and $G'S$ are equal, and
Lemma~\ref{L-CyclicMod} provides us with an invertible matrix $C\in\cU(S)$ such that $G'=GC$.
Thus, $f(x)=f(uG)=uGA=uG'=uGC=xC$ for all $x=uG\in\cC$, and hence~$f$ is a
global $\cU(S)$-map.
\end{proof}

The following class of rings will turn out to be crucial.

%%%%%%%%%%%%%%%%%%%%%%
\begin{defi}\label{D-Constr}
Let~$S$ be a subring of $R^{n\times n}$ and~$\cU(S)$ be the group of invertible matrices in~$S$.
The subring $S$ is called \emph{constructible} if for all matrices $A_1,\ldots, A_n\in\cU(S)$, the matrix~$B\in R^{n\times n}$
defined as $B=(A_1 e_1\T,\ldots,A_n e_n\T)$ is in~$S$.
\end{defi}
%%%%%%%%%%%%%%%%%%
As in the proof of Theorem~\ref{T-ltnr}, the $i$-th column of the matrix~$B$ is the $i$-th column of~$A_i$.
Observe that we do not require that~$B$ be in $\cU(S)$.

Now we can generalize the proof of Theorem~\ref{T-ltnr} to constructible rings and their groups of units.
For later use we first provide the following more technical result.

%%%%%%%%%%%%%%%%%%
\begin{theo}\label{T-ConstrLocGlob}
Let~$S$ be a constructible subring of $R^{n\times n}$.
Let $f:\cC\longrightarrow\cC'$ be an isomorphism of the codes $\cC,\,\cC'\subseteq R^n$ with the property that for each
$x\in\cC$ there exists a matrix
$A_x\in S$ such that $f(x)=xA_x$ and, likewise, for each $y\in \cC'$  there exists a matrix
$B_y\in S$ such that $f^{-1}(y)=yB_y$.
Then~$f$ is a global $\cU(S)$-map.
\end{theo}
%%%%%%%%%%%%%%%%%%%

\begin{proof}
We first show that $f$ is a local $\cU(S)$-map.
Fix $x\in\cC$ and let $y=f(x)$. Then $y=xA_x$ and $x=yB_y$ for some matrices $A_x,\,B_y\in S$, and thus
the right $S$-modules $xS$ and $yS$ coincide.
Lemma~\ref{L-CyclicMod} yields $y=xC_x$ for some $C_x\in\cU(S)$, which
shows that~$f$ is a local $\cU(S)$-map.
Consequently, $f^{-1}$ is a local $\cU(S)$-map as well.

In order to show that~$f$ is a global $\cU(S)$-map, we proceed as in the proof of Theorem~\ref{T-ltnr}.
By Lemma~\ref{L-technical} there exists for each standard basis vector $e_i\in R^n$ a matrix $A_i\in\cU(S)$ such that
$f_i(x)=\inner{f(x),e_i}=xA_i e_i\T$ for all $x\in\cC$.
Define the matrix $A\in R^{n\times n}$ as $A=(A_1 e_1\T,\ldots,A_n e_n\T)$.
By constructibility, the matrix~$A$ is in~$S$.
Furthermore,
\begin{equation}\label{e-fxA}
    f(x)=xA\text{ for all }x\in\cC.
\end{equation}
In the same way we obtain a matrix $B\in S$ such that $f^{-1}(y)=yB$ for all $y\in\cC'$.
Now Lemma~\ref{L-Sunits} concludes the proof.
\end{proof}

The last result along with the fact that the inverse of any local $\cU(S)$-map is also a local $\cU(S)$-map immediately leads to the
local-global property of $\cU(S)$.
%%%%%%%%%%%%%%%%%%%%%
\begin{theo}\label{T-Constr}
Let~$S$ be a constructible subring of $R^{n\times n}$.
Then~$\cU(S)$ has the local-global property.
\end{theo}
%%%%%%%%%%%%%%%%%%%

Constructibility of the ring~$R^{n\times n}$ and of the ring of diagonal matrices leads to the following.

%%%%%%%%%%%%%%%%%%%%%
\begin{cor}\label{C-GLDelta}
The general linear group $\glnr$ and the group $\Delta(n,R)$ of invertible diagonal matrices satisfy the local-global property.
\end{cor}
%%%%%%%%%%%%%%%%%%%%%%
We have seen already in Example~\ref{E-GLLocNonGlob} that the local-global property of~$\glnr$ does not hold in general if the
ring~$R$ is not Frobenius.

Further groups satisfying the local-global property are obtained with the following simple properties.

%%%%%%%%%%%%%%%%%%
\begin{rem}\label{R-Intersection}
If $S_1,\,S_2\subseteq R^{n\times n}$ are constructible subrings, then so is $S_1\cap S_2$.
As a consequence, $\cU(S_1)\cap\cU(S_2)$ satisfies the local-global property.
\end{rem}
%%%%%%%%%%%%%%%%%%%

%%%%%%%%%%%%%%%%%%%
\begin{prop}\label{P-Conj}
Let $\cU$ be a subgroup of $\glnr$ that satisfies the local-global property.
Then each conjugate $\cU^P:=P\cU P^{-1}$, where $P\in\glnr$, also satisfies the local-global property.
\end{prop}
%%%%%%%%%%%%%%%%%%%%
\begin{proof}
Let $f:\cC\longrightarrow\cC'$ be a local $\cU^P$-map between the codes $\cC,\,\cC'\subseteq R^n$.
Then for any $x\in \cC$ there is an $A_x\in\cU$ such that $f(x)=x(PA_xP^{-1})$.
It follows that $f(x)P=xPA_x$ for all $x\in\cC$.
Define the map $g:\cC P\longrightarrow \cC'P$ by $g(xP):=f(x)P$.
Notice that~$g$ is linear and satisfies $g(xP)=(xP)A_x$. Hence $g$ is a local $\cU$-map.
By the local-global property of~$\cU$ there exists a matrix $A\in\cU$ such that $g(xP)=xPA$ for all $x\in \cC$.
Now $f(x)=g(xP)P^{-1}=xPAP^{-1}$ for all $x\in\cC$, showing that~$f$ is a global $\cU^P$-map.
\end{proof}

The following result is immediate from the definition of constructibility.

%%%%%%%%%%%%%%%%%%
\begin{theo}\label{T-triangular}
Let $S_i$ be constructible subrings of $R^{n_i\times n_i}$ for $i=1,\ldots,t$. Put $n:=n_1+\ldots + n_t$.
\begin{alphalist}
\item The subring $\mathrm{diag}(S_1,\ldots,S_t)$ of $R^{n\times n}$ is constructible.
       As a consequence, the subgroup of $\glnr$ consisting of all block diagonal matrices of the form
      $\mathrm{diag}(A_1,\ldots, A_t)$, where $A_i\in\cU(S_i)$, satisfies the local-global property.
\item  The subring of $R^{n\times n}$ consisting of all matrices of the form
      \begin{equation}\label{e-blocktriang}
	      \begin{pmatrix}
	      A_{11}&&&\\
	      A_{21}&A_{22}&&\\
	      \vdots&\vdots&\ddots\\
	      A_{t1}&A_{t2}&\cdots&A_{tt}
	      \end{pmatrix},
      \end{equation}
      where $A_{ii}\in S_i$  and $A_{ij}$ is any matrix in $R^{n_i\times n_j}$, is constructible, and thus
      the subgroup of $R^{n\times n}$ consisting of all lower block triangular matrices as
      in~\eqref{e-blocktriang} and where $A_{ii}\in\cU(S_i)$ satisfies the local-global property.
\end{alphalist}
\end{theo}
%%%%%%%%%%%%%%%%%%%%

While part~(a) is an obvious, and actually not very helpful result (see also Remark~\ref{R-ProdWeight} in the next section), we will encounter
an interesting and non-trivial application of part~(b) in Section~\ref{SS-Poset}.

\medskip
We have not yet addressed the local-global property of the most prominent group in the area of MacWilliams extension theorems: the group of monomial matrices.
Since this group is not the group of units of a constructible ring, the local-global property does not follow from the preceding
considerations.
Even though the MacWilliams extension theorem is well-known for monomial matrices~\cite{Wo97}, we think it is worthwhile to briefly
discuss this case in the language of this paper.

Let $U\leq R^\times$ be a subgroup of the multiplicative group of~$R$ and let $\MonUR\leq\GL(n,R)$ be as in Definition~\ref{D-MonMat}.
The local $\MonUR$-maps can be described in the following way.
Let $\cP_U=(P_\ell)_{\ell=1}^t$, that is, $P_1,\,\ldots,P_t$ are the distinct orbits of the right action~$\rho_r$ of~$U$ on~$R$ as in Proposition~\ref{P-Uorbits}.
For $x\in R^n$ define the composition vector as $\text{comp}_{U}(x):=(s_1,\ldots,s_t)$, where $s_\ell=|\{i\mid x_i\in P_\ell\}|$.
Then a linear map~$f:\cC\longrightarrow R^n$ is a local $\MonUR$-map if and only if it is $\text{comp}_{U}$-preserving, that is,
$\text{comp}_{U}\big(f(x)\big)=\text{comp}_{U}(x)$ for all $x\in\cC$.

For instance, if $U$ is the trivial subgroup $U=\{1\}$,
then $\text{comp}_{U}(x)$ is the \emph{complete weight} of~$x$, defined as
$\text{cw}(x)=(s_\alpha\mid \alpha \in R)$, where $s_\alpha=|\{i\mid x_i=\alpha\}|$; see for instance~\cite[p.~142]{MS77}.
On the other hand, if $R$ is a field and $U=R^\times$, then $\text{comp}_{U}(x)$ is, up to notation, the Hamming weight of~$x$.
Thus, the $\text{comp}_{U}$-preserving maps are exactly the Hamming weight-preserving maps.
It is well known, but not obvious, that the same is true for Frobenius rings; see Wood~\cite[Thm.~6.3]{Wo99}.
We will encounter that result in Theorem~\ref{T-MacWExtHamm}.

The following theorem has been coined the MacWilliams extension theorem for $\text{comp}_{U}$-pre\-ser\-ving maps.
In that terminology it has been shown by Wood~\cite[Thm.~10]{Wo97}; see also Theorem~\ref{T-Goldberg} by Goldberg.
For later use we provide a proof in our language.
It is slightly shorter than the one in~\cite{Wo97} because our approach allows us to circumvent the argument based on
averaging characters needed in~\cite[Thm.~10]{Wo97}.

%%%%%%%%%%%%%%%%%
\begin{theo}\label{T-MonUR}
For any subgroup $U\leq R^\times$ the group $\MonUR$ satisfies the local-global property.
\end{theo}
%%%%%%%%%%%%%%%%%%

\begin{proof}
Let $\cU:=\MonUR$ and $f:\cC\longrightarrow R^n$ be a local $\cU$-map.
The group $\text{Hom}_R(\cC,R)$ of left $R$-linear maps is a right $R$-module by defining
$(f\!\cdot\! \alpha)(x)=f(x)\alpha$ for all $x\in\cC$.
Denote by $f_1,\ldots,f_n$ the coordinate functions of~$f$ and by $\pi_1,\ldots,\pi_n$ the projections of~$R^n$ on its components.
Then $f_i,\,\pi_i\in\text{Hom}_R(\cC,R)$.
We have to show that there exist $\alpha_1,\ldots,\alpha_n\in U$ and a permutation $\tau\in S_n$ such that
$f_i=\pi_{\tau(i)}\!\cdot\!\alpha_i,\,i=1,\ldots,n$.

Let $Q$ be the block of $\widehat{\cP_{\cU}}^{\scriptscriptstyle[\chi,l]}=\cP_{\cU^{\sf T}}$ that contains the standard basis vector~$e_1$.
Then the very definition of~$\cP_{\cU^{\sf T}}$ yields $Q=\{\alpha e_i\mid i=1,\ldots,n,\,\alpha\in U\}$.
Lemma~\ref{L-technical} yields
\begin{equation}\label{e-sumoverP}
    \sum_{y\in Q}\chi(\langle f(-),y\rangle)=\sum_{y\in Q}\chi(\langle -,y\rangle)
\end{equation}
as sums of characters on~$\cC$.
Thanks to Proposition~\ref{P-LinIndepChar}, the character $\chi(\inner{f(-),e_1})$ must appear on the right hand side
of~\eqref{e-sumoverP}.
Thus there exists some $\alpha_1\in U$ and $e_{\tau(1)}$ such that $\chi(\inner{f(-),e_1})=\chi(\inner{-,\alpha_1e_{\tau(1)}})$
as characters on~$\cC$, and Proposition~\ref{P-GenChar}(2) yields that the linear maps $\inner{f(-),e_1}$ and
$\inner{-,\alpha_1e_{\tau(1)}}$ coincide on~$\cC$.
In other words $f_1= \pi_{\tau(1)}\!\cdot\!\alpha_1$.

Next, the set~$Q$ contains the sets $Q_i:=\{\alpha e_i\mid \alpha \in U\}$ for each $i=1,\ldots,n$.
Moreover,
$ \sum_{y\in Q_1}\chi(\inner{f(-),y})=\sum_{\alpha \in U} \chi(\inner{f(-), \alpha e_1})
  =\sum_{\alpha \in U} \chi(\inner{f(-), e_1} \alpha)
   =\sum_{\alpha \in U} \chi(\inner{-,\alpha_1 e_{\tau(1)}}\alpha )
   =\sum_{y\in Q_{\tau(1)}}\chi(\inner{-,y})$.
Consequently, the identity~\eqref{e-sumoverP} can be reduced to
\begin{equation}\label{e-sumoverPP}
    \sum_{y\in Q\backslash Q_1}\chi(\inner{f(-),y})=\sum_{y\in Q\backslash Q_{\tau(1)}}\chi(\inner{-,y}).
\end{equation}
In other words, we eliminated all scalar multiples of $e_1$ and $e_{\tau(1)}$ on the left and right hand side
of~\eqref{e-sumoverP}, respectively.
Now we may repeat the argument with the character $\chi(\inner{f(-),e_2})$ appearing on the left hand side
of~\eqref{e-sumoverPP}.
Again, it must appear on the right hand side, and as above this means that there exists some $\alpha_2\in U$ and $\tau(2)\neq\tau(1)$ such that %$\chi(\inner{e_2,f(-)})=\inner{\alpha_2e_{\tau(2)},-}$.
$f_2=\inner{-,\alpha_2 e_{\tau(2)}}=\pi_{\tau(2)}\!\cdot\!\alpha_2$.

Continuing in this fashion, we obtain the desired result.
\end{proof}

%%%%%%%%%%%%%%%%%%%%%%%%%%%%%%%%%%%%%%%%%%
\section{Weight-Preserving Maps and MacWilliams Extension Theorems}\label{SS-Goldberg}
%%%%%%%%%%%%%%%%%%%%%%%%%%%%%%%%%%%%%%%%%%%
This section is devoted to establishing MacWilliams extension theorems for various isometries.
More precisely, we address Question~\ref{Q-Goldberg} and show that for certain weight functions, the weight-preserving
isomorphisms between codes in~$R^n$ are given by global $\cU$-maps for a suitable group~$\cU\leq\GL(n,R)$, and thus extend to
weight-preserving isomorphisms on~$R^n$.

As before, let~$R$ be a finite Frobenius ring.

\medskip
We begin with the Rosenbloom-Tsfasman weight (RT-weight) on $R^n$, see~\cite{RoTs97}.
For a vector $x=(x_1,\ldots,x_n)\in R^n$  the RT-weight is defined as
\begin{equation}\label{e-RTweight}
  \wtRT(x):=\left\{\begin{array}{cl}0,&x=0,\\ \max\{i\mid x_i\neq 0\},&\text{otherwise.}\end{array}\right.
\end{equation}
It is not hard to see that the distance between two vectors $x,y$ defined as ${\rm d}_{\rm RT}(x,y):=\wtRT(x-y)$ is a metric
on $R^n$ (see also \cite[Lem.~1.1]{BGL95}).

The RT-metric is, just like the Hamming metric, a special case of a poset metric.
The latter will be studied in greater generality in the next section where also some background is provided.

Evidently, the RT-weight is constant on the $\ltnr$-orbits (but vectors with the same RT-weight need not be in the same orbit,
e.g., $(1,0),\,(2,0)\in\Z_4^2$).
In the following theorem we show that every $\wtRT$-preserving isomorphism $f:\cC\longrightarrow R^n$ is a global $\ltnr$-map.
For the case where $\cC=\F^n$ the result appears as a special case in~\cite[Thm.~1]{Lee03}; see also Remark~\ref{R-CumRT} below.

%%%%%%%%%%%%%%%%%%%%%%%%%%%%%%%%
\begin{theo}\label{T-rho}
Let $f:\cC\longrightarrow R^n$ be a linear $\wtRT$-preserving map on the code $\cC\subseteq R^n$, thus $\wtRT(f(x))=\wtRT(x)$ for all $x\in\cC$.
Then $f$ is a global $\ltnr$-map.
As a consequence, the $\wtRT$-preserving maps satisfy the MacWilliams extension theorem, and
the group of isometries of the metric space $(R^n,{\rm d}_{\rm RT})$ is given by $\ltnr$.
\end{theo}
%%%%%%%%%%%%%%%%%%%%%%%%%%%%%%%%

The last statement, concerning the isometry group of $(R^n,{\rm d}_{\rm RT})$, can also easily be seen by considering the standard basis of~$R^n$.

The examples in~\ref{E-GLLocNonGlob} and~\ref{E-NonFrobExa} show that the result is not true for non-Frobenius rings.
In that case a $\wtRT$-preserving map may not even be a \emph{local} $\ltnr$-map.

In the next section we will provide a much more general result establishing the extension property of certain poset-weight-preserving isomorphisms.
Since its proof will be considerably more elaborate, we think it is worth presenting a direct proof of the above case.

\begin{proof}[Proof of Theorem~\ref{T-rho}]
Let $f:\cC\longrightarrow R^n$ be a $\wtRT$-isometry.
Denote the ring of all lower triangular matrices by $S\subseteq R^{n\times n}$.
Then~$S$ is constructible and $\ltnr=\cU(S)$, the group of invertible matrices in~$S$.\footnote{We thank Jay Wood for pointing out the
interesting fact that~$S$ is not a Frobenius ring; see for instance~\cite[Ex.~(15.26)]{Lam99}.}
Note that both~$f$ and~$f^{-1}$ (defined on the image of~$f$) are $\wtRT$-isometries.
Therefore, using symmetry and Theorem~\ref{T-ConstrLocGlob}, it suffices to show that for every $x\in\cC$
there exists a matrix $A_x\in S$ such that $f(x)=xA_x$.

Fix $x\in\cC$ and let $y=f(x)$.
The existence of a matrix~$A_x$ is equivalent to the existence of elements
$a_{ij}\in R,\,i\geq j,$ such that $y_j=\sum_{i=j}^{n} x_ia_{ij}$.
This is equivalent to showing that $y_j\in x_jR+\ldots +x_nR$ for all~$j$.
We make use of annihilator ideals.
Let $r \in \mathrm{ann}_l(x_jR+\ldots +x_nR)$. Then $r x_i=0$ for $i\geq j$ and therefore the vector $rx$ satisfies
 $\wtRT(r x)< j$.
Thus $\wtRT(f(rx))=\wtRT(r y)< j$.
In particular $r y_j=0$.
All of this shows that $\mathrm{ann}_l(x_jR+\ldots+x_nR)\subset \mathrm{ann}_l(y_jR)$.
Now the double annihilator property of ideals, see Remark~\ref{R-FrobAnn}, yields
\[
  y_jR=\mathrm{ann}_r(\mathrm{ann}_l(y_jR))\subset \mathrm{ann}_r(\mathrm{ann}_l(x_jR+\ldots +x_nR))
         =x_jR+\ldots + x_nR,
\]
and therefore $y_j\in x_jR+\ldots +x_nR$.
All of this establishes the existence of a matrix~$A_x\in S$ such that $f(x)=xA_x$.
This concludes the proof.
\end{proof}

%%%%%%%%%%%%%%%%%%%%%%%%%%%%
\begin{rem}\label{R-CumRT}
In~\cite[Thm.~1]{Lee03} Lee considers the Niederreiter-Rosenbloom-Tsfasman (NRT) space $\F^{t\times n}$, where for any matrix $M\in\F^{t\times n}$
the NRT-weight is defined as $\wtNRT(M)=\sum_{i=1}^t\wtRT(M_i)$, where $M_i$ is the $i$-th row of~$M$.
Lee proved that the group of NRT-weight-isometries on $\F^{t\times n}$ is given by a semidirect product of $(\ltnr)^t$ and the permutation group~$S_t$.
It is easy to see that NRT-weight-preserving maps do not satisfy the MacWilliams extension theorem:
consider the two binary codes
\[
  \cC=\Big\{\begin{pmatrix}0&0\\0&0\end{pmatrix},\,\begin{pmatrix}1&0\\1&0\end{pmatrix}\Big\},\
  \cC'=\Big\{\begin{pmatrix}0&0\\0&0\end{pmatrix},\,\begin{pmatrix}0&0\\0&1\end{pmatrix}\Big\}\subseteq\F_2^{2\times2}.
\]
The (unique) linear isomorphism between these codes is NRT-weight-preserving, but cannot be extended to an NRT-isometry on $\F_2^{2\times2}$.
This follows directly from linearity along with the fact that the set $\big\{\Smallfourmat{1}{0}{0}{0},\,\Smallfourmat{0}{0}{1}{0}\big\}$ has to be
mapped to itself under any NRT-isometry.
We will come back to this situation in Example~\ref{E-PosetNonExt} in the context of poset weights.
\end{rem}
%%%%%%%%%%%%%%%%%%%%%%%%%%%%

In a similar way as in Theorem~\ref{T-rho}, we can deal with support-preserving maps.
For a vector $x\in R^n$ let $\supp(x)=\{i\mid x_i\neq0\}$.
Clearly the support is constant for all~$x$ in the same $\Delta(n,R)$-orbit,
where $\Delta(n,R)$ is the group of invertible diagonal matrices over~$R$.
On the other hand, vectors with the same support need not be in the same  $\Delta(n,R)$-orbit.
Now we can prove the MacWilliams extension theorem for support-preserving maps.
Again, Examples~\ref{E-GLLocNonGlob} and~\ref{E-NonFrobExa} show that it is not true for non-Frobenius rings.

%%%%%%%%%%%%%%%%%%%
\begin{theo}\label{T-supp}
Let $f:\cC\longrightarrow \cC'$ be a support-preserving linear isomorphism, i.e., $\supp(f(x))=\supp(x)$ for all $x\in\cC$.
Then~$f$ is a global $\Delta(n,R)$-map.
In particular, the group of support-preserving isomorphisms on~$R^n$ is given by $\Delta(n,R)$.
\end{theo}
%%%%%%%%%%%%%%%%%%%%

\begin{proof}
Again we make use of the double annihilator property of $R$.
Let $x\in\cC$ and fix $j\in\{1,\ldots,n\}$.
Denote by~$f_j$ the $j$-th coordinate function of~$f$.
Since $f$ is linear and preserves the support we have $\alpha x_j=0$ if and only if $\alpha f_j(x)=0$
for all $\alpha\in R$.
This implies $\mathrm{ann}_l(x_jR)=\mathrm{ann}_l(f_j(x)R)$,
and with the aid of Remark~\ref{R-FrobAnn} we obtain $x_jR=f_j(x)R$.
Hence Lemma~\ref{L-CyclicMod} yields the existence of $\alpha_j\in R^\times$ such that $f_j(x)= x_j\alpha_j$.
Since~$j$ is arbitrary, all of this shows that $f(x)=x\cdot\mathrm{diag}(\alpha_1,\ldots, \alpha_n)$.
Thus~$f$ is a local $\Delta(n,R)$-map, and Corollary~\ref{C-GLDelta}  concludes the proof.
\end{proof}

Now we turn to the group $\glnr$. The following result establishes a criterion for when an isomorphism $f:\cC\longrightarrow \cC'$
can be extended to an isomorphism on $R^n$.
The reader may wish to compare this situation with that of injective modules.
The ring~$R$, being a finite Frobenius ring, is an injective (left) $R$-module~\cite[Thm.~(15.1)]{Lam99}, and hence
so is~$R^n$~\cite[Prop.~(3.4)]{Lam99}.
Therefore, by definition of injectivity, every linear map $f:\cC\longrightarrow R^n$ can be extended to a linear map
$\hat{f}: R^n\longrightarrow R^n$.
The criterion below characterizes the case where there exists even an \emph{isomorphism}~$\hat{f}$ extending~$f$.

%%%%%%%%%%%%%%%%%
\begin{theo}\label{T-idealpreserving}
Let $f:\cC\longrightarrow \cC'$ be a linear map between codes $\cC,\,\cC'\subseteq R^n$.
Then the following are equivalent.
\begin{alphalist}
\item $f$ is a global $\glnr$-map. In other words,~$f$ can be extended to an isomorphism on~$R^n$.
\item  For all $x\in\cC$ and $y=f(x)\in\cC'$, the right ideals in~$R$ generated by the entries of~$x$ and~$y$,
         respectively, coincide, i.e., $\sum_{i=1}^n y_iR=\sum_{i=1}^n x_iR$.
\end{alphalist}
As a consequence, the linear right-ideal-preserving maps satisfy the MacWilliams extension theorem, and
the group of right-ideal-preserving isomorphisms on~$R^n$ is given by $\glnr$.
\end{theo}
%%%%%%%%%%%%%%%%%
Observe that~(b) implies injectivity of~$f$.
Again, Example~\ref{E-GLLocNonGlob} shows that the above equivalence is not true if~$R$ is not Frobenius.

\begin{proof}
The implication (a)~$\Rightarrow$~(b) is clear because $f(x)=xM$ for some $M\in\glnr$.
For the converse let $x\in\cC$ and put $y=f(x)$.
Then $y_j\in \sum_{i=1}^n x_iR$ for all $j=1,\ldots,n$, and this gives rise to a matrix
$A_x\in R^{n\times n}$ such that $y=xA_x$.
In the same way we have for all $y\in\cC'$ a matrix $B_y$ such that $f^{-1}(y)=yB_y$.
Now the constructibility of $R^{n\times n}$ along with Theorem~\ref{T-ConstrLocGlob} yields~(a).
\end{proof}

We have not yet discussed the most famous MacWilliams extension theorem, that is, the one for Hamming weight-preserving maps.
Note that if $f:\;\cC\longrightarrow\cC'$ is a Hamming weight-preserving linear map between codes in~$\F^n$, where~$\F$ is field, then
it is immediate that~$f$ is a local $\MonF$-map, and thus global due to Theorem~\ref{T-MonUR}.
This way, Theorem~\ref{T-MonUR} yields another proof of the classical MacWilliams extension theorem for codes over fields endowed
with the Hamming weight. (This is exactly the line of reasoning by Goldberg~\cite[p.~364]{Gol80}.)
If, however,~$R$ is not a field, then it is not a priori clear whether a Hamming weight-preserving linear map~$f$ is even a
local $\MonR$-map since distinct nonzero elements in~$R$ need not be equal up to a unit factor.
As a consequence, the Hamming partition of~$R^n$ is not given by the orbits of a suitable group action.
Nevertheless, the MacWilliams extension theorem remains true, as has been been proven by Wood~\cite{Wo99}.

%%%%%%%%%%%%%%%%%%%%%%%%%%%%%%%%%%%%%%%%%%%
\begin{theo}[\mbox{\cite[Thm.~6.3]{Wo99}}]\label{T-MacWExtHamm}
Every Hamming weight-preserving linear map between codes in~$R^n$ is a global $\MonR$-map.
\end{theo}
%%%%%%%%%%%%%%%%%%%%%%%%%%%%%%%%%%%%%%%%%

Note that this result implies Theorem~\ref{T-MonUR} for the case where $U=R^\times$.

Since we will make use of this result in the next section, we present a short proof in the terminology of our paper.
It shows that one establishes directly that a Hamming weight-preserving linear map is a global $\MonR$-map without first showing that it is a
local $\MonR$-map.
The proof is but a slight adjustment of the one for Theorem~\ref{T-MonUR}.

\begin{proof}
Let $f:\cC\longrightarrow\cC'$ be a Hamming weight-preserving map between codes $\cC,\,\cC'\subseteq R^n$.
Consider the notation of the proof of Theorem~\ref{T-MonUR}, and in particular the right $R$-module structure of $\text{Hom}_R(\cC,R)$.
Again, we have to show that there exist $\alpha_1,\ldots,\alpha_n\in R^\times$ and a permutation $\tau\in S_n$ such that
$f_i=\pi_{\tau(i)}\!\cdot\! \alpha_i$ for $i=1,\ldots,n$.

In the proof of Theorem~\ref{T-MonUR} replace the set~$Q$ by
$Q:=\{\alpha e_i\mid i=1,\ldots,n,\,\alpha\in R\backslash\{0\}\}$, which is exactly the set of vectors in~$R^n$ with Hamming
weight one.
Since the Hamming partition~$\cP$ satisfies $\cP=\wcPchir$, see Example~\ref{E-HamPart}, and~$x$ and~$f(x)$ have the same Hamming
weight,~\eqref{e-simPhat2r} implies that~\eqref{e-sumoverP} remains true in this situation.

For the main argument of the proof, we proceed as follows.
Without loss of generality we may assume that the submodule~$f_1\!\cdot\! R$ is maximal among the submodules
$f_1\!\cdot\! R,\ldots,f_n\!\cdot\! R,\,\pi_1\!\cdot\! R,\ldots, \pi_n\!\cdot\! R$ of $\text{Hom}_R(\cC,R)$,
i.e., $f_1\!\cdot\! R$ is not properly contained in any of these modules (the situation is symmetric with respect to
$f_i$ and $\pi_i$ because~$\pi_i$ is the $i$-th coordinate function of the identity map on~$\cC$).
Choose $y=e_1\in Q$.
As in the proof of Theorem~\ref{T-MonUR}, the character $\chi(\inner{f(-),e_1})=\chi\circ f_1$ has to appear
on the right hand side of~\eqref{e-sumoverP}.
Hence there exist $\beta\neq0$ and $\tau(1)\in\{1,\ldots,n\}$ such that
$\chi\circ f_1=\chi(\inner{-,\beta e_{\tau(1)}})=\chi\circ\pi_{\tau(1)}\!\cdot\!\beta$, and
Proposition~\ref{P-GenChar}(2) implies $f_1=\pi_{\tau(1)}\!\cdot\!\beta$.
Thus $f_1\!\cdot\! R\subseteq  \pi_{\tau(1)}\!\cdot\! R$, and from the maximality of~$f_1\!\cdot\! R$ we obtain
$f_1\!\cdot\!R= \pi_{\tau(1)}\!\cdot\!R$.
Therefore Lemma~\ref{L-CyclicMod} yields the existence of a unit
$\alpha_1\in R^\times$ such that $f_1=\pi_{\tau(1)}\!\cdot\!\alpha_1$.

Now we may proceed as in the proof of Theorem~\ref{T-MonUR} and
reduce the identity~\eqref{e-sumoverP} to~\eqref{e-sumoverPP}, where the sets~$Q_j$ are defined as
$Q_j:=\{\alpha e_j\mid \alpha\neq0\}$.
\end{proof}

The proof of Theorem~\ref{T-MonUR} lends itself for a very short proof of the well-known fact that the global $\MonR$-maps are
exactly the linear maps $\cC\longrightarrow R^n$ that preserve the homogeneous weight. Let us briefly elaborate on this.

Recall that a \emph{(left) homogeneous weight on~$R$ with average value~$\gamma$} is a function
$\omega:R\longrightarrow\Q$ such that $\omega(0)=0$ and
\begin{romanlist}
\item $\omega(x)=\omega(y)$ for all $x,y\in R$ such that $Rx=Ry$,
\item $\sum_{y\in Rx}\omega(y)=\gamma|Rx|$ for all $x\in R\backslash\{0\}$; in other words, the average weight over each nonzero principal left ideal is~$\gamma$.
\end{romanlist}

In~\cite[Thm.~1.3]{GrSch00} Greferath/Schmidt establish existence and uniqueness of the homogeneous weight
(even for arbitrary finite rings).

With the aid of a generating character~$\chi$, Honold~\cite[p.~412]{Hon01} derived the following explicit
formula for the normalized (i.e., $\gamma=1$) homogeneous weight:
\begin{equation}\label{e-HomWt}
   \omega(r)=1-\frac{1}{|R^\times|}\sum_{u\in R^\times}\chi(ru) \text{ for } r\in R.
\end{equation}
For more on the explicit values of the homogeneous weight in Frobenius rings, see also~\cite{GL13homog,GL14homog}.

Extend now the homogeneous weight additively to~$R^n$, that is
$\omega(x_1,\ldots,x_n)=\sum_{i=1}^n\omega(x_i)$ for all $(x_1,\ldots,x_n)\in R^n$.\footnote{Note that the
resulting weight is different from the homogeneous weight on the ring $R\times\ldots\times R$.}
It is clear that $\MonR$-maps preserve the homogeneous weight on~$R^n$.
In~\cite{GrSch00}  Greferath/Schmidt prove the converse.
%%%%%%%%%%%%%%%%%%%%%%
\begin{theo}[\mbox{\cite[Thm.~2.5]{GrSch00}}]
Let $f:\cC\longrightarrow\cC'$ be an isomorphism that preserves the homogeneous weight, i.e.,
$\omega(x)=\omega(f(x))$ for all $x\in\cC$.
Then $f$ is a global $\MonR$-map.
As a consequence, the homogeneous weight-preserving maps satisfy the MacWilliams extension theorem.
Furthermore, an isomorphism $f:\cC\longrightarrow\cC'$ preserves the homogeneous weight if and only if it preserves the Hamming weight.
\end{theo}
%%%%%%%%%%%%%%%%%%%%%%%

The proof in~\cite{GrSch00} is purely combinatorial and does not make use of character theory.
We think it is worthwhile to present the following alternative and very short proof.

\begin{proof}
Due to~\eqref{e-HomWt}, the normalized homogeneous weight of $x\in R^n$ is given by
$ \omega(x)=\sum_{i=1}^n\omega(x_i)=n-\frac{1}{|R^\times|}\sum_{i=1}^n\sum_{u\in R^\times}\chi(x_iu)$.
Therefore
$\omega(x)=\omega(x')\Longleftrightarrow
\sum_{y\in Q}\chi(\inner{x,y})=\sum_{y\in Q}\chi(\inner{x',y})$ for the set
$Q:=\{u e_i\mid u\in R^\times,\,i=1,\ldots,n\}$.
For a linear map~$f$ preserving the homogeneous weight, this reads as
$\sum_{y\in Q}\chi(\inner{f(-),y})=\sum_{y\in Q}\chi(\inner{-,y})$.
But this is exactly~\eqref{e-sumoverP} for the subgroup $U=R^\times$, and thus the proof of Theorem~\ref{T-MonUR} shows that~$f$
is a global $\MonR$-map.
\end{proof}

\bigskip
The rest of this section is devoted to discussing whether the results may be extended to the case where the ambient space is a direct
product $R^N:=R^n\times\ldots\times R^n$ and where~$R^n$ is endowed with a particular weight function.
In this case, we consider two types of weights on~$R^N$.

%%%%%%%%%%%%%%%%%%%%%%%%%%%%%%%
\begin{defi}\label{D-DirProd}
Let $N=tn$ and $R^{N}:=R^n\times\ldots\times R^n$,  and suppose~$R^n$ is endowed with a weight~$\wt$.
For $x=(x_1,\ldots,x_t)\in R^N$ we define the \emph{weight list} of~$x$ as the list $\big(\wt(x_1),\ldots,\wt(x_t)\big)$
and the \emph{symmetrized weight composition} as the multi-set $\{\!\{\wt(x_1),\ldots,\wt(x_t)\}\!\}$.
In other words, the symmetrized weight composition is the list of weights of the components regardless of order.
\footnote{A third type of weight on~$R^N$
is the sum of the weights of the components.
Except for the homogeneous weight and certain cases covered by poset weights, to be dealt with in the next section,
such cumulative weights require completely different methods and will not be discussed here; see for instance Wood~\cite{Wo01} and~\cite{Bar12}
for the still open problem whether maps preserving the cumulative Lee weight on~$(\Z_M)^n$ satisfy the extension theorem, as well as
Greferath et al.~\cite{GMZ13} for a more recent approach to cumulative weights.}
\end{defi}
%%%%%%%%%%%%%%%%%%%%%%%%%%%%%%

For example, let $n=1$ and~$\wt$ be the Hamming weight on~$R$ (thus, $\wt(a)=1$ for $a\neq0$ and $\wt(0)=0$).
Then the weight list of $x\in R^N$ equals its support vector (i.e., the unique vector in $\{0,1\}^N$ whose support coincides
with the support of~$x$).
Thus a weight-list-preserving map is simply a support-preserving map.
On the other hand, the symmetrized weight composition is, up to notation, the Hamming weight of~$x$, and a map preserving the
symmetrized weight composition is a Hamming weight-preserving map.
We know from Theorems~\ref{T-supp} and~\ref{T-MacWExtHamm} that in both cases the MacWilliams extension theorem is true, which also implies that the
according groups $\Delta(N,R)$ and $\Mon(N,R)$ satisfy the local-global property.

This motivates the following question for the general situation of Definition~\ref{D-DirProd}.
Suppose the weight~$\wt$ on~$R^n$ satisfies the MacWilliams extension theorem.
Does then the MacWilliams extension theorem also hold true for weight-list-preserving maps or
symmetrized-weight-composition-preserving maps?

As we now briefly outline, the first case can be answered in the affirmative, whereas this is not the case for the
symmetrized weight composition.
We begin with the non-symmetrized case, which can even be dealt with in slightly more generality.

%%%%%%%%%%%%%%%%%%%%%%%%%%%%%%%
\begin{rem}\label{R-ProdWeight}
Let $R^{N}:=R^{n_1}\times\ldots\times R^{n_t}$ and suppose that each~$R^{n_i}$ is endowed with a weight $\wt_i$
such that $\wt_i(x)=0\Longleftrightarrow x=0$ for all $x\in R^{n_i}$ and that satisfies the
MacWilliams extension theorem.
Furthermore, for each~$i$ let $\cU_i\leq\GL(n_i,R)$ be the group of matrices inducing the $\wt_i$-preserving isomorphisms on~$R^{n_i}$.
Denote by $\pi_i:\,R^N\longrightarrow R^{n_i}$ the projection of~$R^{N}$ on the $i$-th component.
Let $\cC\subseteq R^N$ be a code and $f:\cC\longrightarrow R^N$ be a linear map.
Assume~$f$ preserves $(\wt_1,\ldots,\wt_t)$, i.e.,
$\big(\wt_1(x_1),\ldots,\wt_t(x_t)\big)=\big(\wt_1(f_1(x)),\ldots,\wt_t(f_t(x))\big)$ for all $x=(x_1,\ldots,x_t)\in\cC$, where
$f_i:=\pi_i\circ f$ is the $i$-th block coordinate map of~$f$.
Then~$f$ is a global $\cU$-map, where $\cU\leq\GL(N,R)$ is the group of block diagonal matrices of the form
\begin{equation}\label{e-Udiag}
   U:=\text{diag}(U_1,\ldots,U_t),\ \text{ where } U_i\in\cU_i\text{ for all }i=1,\ldots,t.
\end{equation}
This can be seen as follows.
The $(\wt_1,\ldots,\wt_t)$-preserving property of~$f$ implies that the maps
\[
  g_i:\pi_i(\cC)\longrightarrow R^{n_i},\ x_i\longmapsto f_i(x), \text{ where $x\in\cC$ is such that $\pi_i(x)=x_i$},
\]
are well-defined, linear, and $\wt_i$-preserving.
Thus by assumption on~$\wt_i$, there exist matrices $U_i\in\cU_i$ such that $g_i(x_i)=x_iU_i$ for all $x_i\in\pi_i(\cC)$.
But then the map~$f$ is given by $f(x)=(x_1,\ldots,x_t)U$, where~$U$ is as in~\eqref{e-Udiag}.
\end{rem}
%%%%%%%%%%%%%%%%%%%%%%%%%%%%%%

Let us now turn to the symmetrized weight composition.
The following simple example shows that in this case one cannot expect a MacWilliams extension theorem in general.
%%%%%%%%%%%%%%%%%%%%%%%%%%%%%%%%%%
\begin{exa}\label{E-RTWeightSymm}
Let $\F=\F_2$ and consider $\F^8=\F^4\times\F^4$ endowed with the symmetrized composition of RT-weights on each component~$\F^4$.
Thus, define $\wt(x_1,x_2):=\{\!\{\wtRT(x_1),\wtRT(x_2)\}\!\}$ for every $(x_1,x_2)\in\F^4\times\F^4$ and where $\wtRT$ is as in~\eqref{e-RTweight}.
In~\cite[Sec.~3.1]{BFSF13} this is called the \emph{shape} of the vector $(x_1,x_2)$.
It is straightforward to show that the $\wt$-preserving isomorphisms on~$\F^8$ are given by the group~$\cU\leq\GL(8,\F)$ consisting of
the $(2\times2)$-block-monomial matrices, where the nonzero blocks are invertible lower triangular $(4\times4)$-matrices.
\\
Consider the codes $\cC_1=\im(G_1)$ and $\cC_2=\im(G_2)$, where
\[
  G_1=\begin{pmatrix}1&0&0&0&0&1&0&0\\0&0&1&0&0&0&0&1\end{pmatrix},\quad
  G_2=\begin{pmatrix}0&1&0&0&1&0&0&0\\0&0&1&0&0&0&0&1\end{pmatrix}.
\]
Evidently, the map $f:\cC_1\longrightarrow\cC_2$ given by $uG_1\longmapsto uG_2$, where $u\in\F^2$,
is $\wt$-preserving (and of course linear).
Moreover, it is easy to see that~$f$ is a local~$\cU$-map.
As we show next,~$f$ cannot be extended to a $\wt$-preserving map~$\hat{f}$ on~$\F^8$.
Assume there exists such an extension~$\hat{f}$ of~$f$.
Then $\hat{f}$ has to map the vector $(1000,0000)$ to $(1000,0000)$ or $(0000,1000)$.
If $\hat{f}(1000,0000)=(1000,0000)$, then linearity of~$\hat{f}$ and the fact that $\hat{f}(1000,0100)=(0100,1000)$
contradicts the $\wt$-preserving property of~$\hat{f}$.
Thus $\hat{f}(1000,0000)=(0000,1000)$ and, using linearity along with the second row of the above matrices, we also have
\begin{equation}\label{e-fhat}
   \hat{f}(1010,0001)=(0010,1001).
\end{equation}
Next,~$\hat{f}$ has to map the vector $(0110,0000)$ to any of the~$8$ vectors~$v$ such that $\wt(v)=\{\!\{0,3\}\!\}$.
But $\hat{f}(0110,0000)=(ab10,0000)$ for any $a,b\in\F_2$ leads to a contradiction to $\hat{f}(1000,0000)=(0000,1000)$, whereas
$\hat{f}(0110,0000)=(0000,ab10)$ contradicts~\eqref{e-fhat}.
All of this also shows that the group~$\cU$ does not satisfy the local-global property.
One can even show that the map~$f:\cC_1\longrightarrow\cC_2$ cannot be extended to a $\wtNRT$-isometry on~$\F^8$, where
$\wtNRT(x_1,x_2)=\wtRT(x_1)+\wtRT(x_2)$ is the NRT-weight (see Remark~\ref{R-CumRT}).
This is an instance of a poset weight, which will be discussed in detail in the next section.

In the same way as above, one can find examples showing that maps preserving the symmetrized composition of Hamming weights do not
satisfy the MacWilliams extension theorem.
\end{exa}
%%%%%%%%%%%%%%%%%%%%%%%%%%%%%%%%

Let us rephrase the above results in terms of the local-global property.
Suppose $U\leq\glnr$.
In the group $\GL(N,R)$, where $N=tn$, define $\Delta_U(N,R)$ as the subgroup consisting of the block diagonal matrices such that the $(n\times n)$-blocks
on the diagonal are matrices in~$U$.
Furthermore, generalizing previous notation, denote by ${\rm Mon}_U(N,R)\leq \GL(N,R)$ the group of $U$-block monomial matrices, that is,
block matrices consisting of~$t$ block rows and~$t$ block  columns such that each block row and each block column contains exactly one nonzero $(n\times n)$-matrix
and this matrix is in~$U$. The reader may note that this group is a semi-direct product of $U^t$ and the symmetric group $S_t$.

The above has shown that if~$U$ satisfies the local-global property then so does $\Delta_U(N,R)$, whereas ${\rm Mon}_U(N,R)$ does not in general.

%%%%%%%%%%%%%%%%%%%%%%%%%%%%%%%%%%%%%%%%%
\section{Isomorphisms Preserving a Poset Weight}\label{SS-Poset}
%%%%%%%%%%%%%%%%%%%%%%%%%%%%%%%%%%%%%%%%%
In this section we generalize Theorem~\ref{T-rho} to general poset-weight-preserving linear maps.
Poset metrics for codes over fields have been introduced by Brualdi et al.~\cite{BGL95} in order to build a suitable framework
for a combinatorial problem posed earlier by Niederreiter~\cite{Nie91} and which generalizes the problem of finding the largest possible
Hamming distance for any $k$-dimensional code in~$\F^n$.
The Hamming metric as well as the RT-metric are special cases of a poset metric.

Ever since~\cite{BGL95}, codes with a poset metric have been studied intensively.
For instance, Skriganov~\cite{Skr07} could establish the existence of good codes with respect to the Hamming distance with the aid of the RT-metric.
In~\cite{DoSk02,KiOh05,PiFi12,GL13pos} MacWilliams identities with respect to a poset weight have been derived.
More closely related to the present paper is~\cite{PFKH08} by Panek et.~al., in which, for an arbitrary poset metric, the group of all isometries
on a given vector space~$\F^n$ is described.
The special case of the NRT-metric, see Remark~\ref{R-CumRT}, has been dealt with already earlier by Lee~\cite{Lee03}.
We will not need these results explicitly, but will point out connections whenever applicable.
A different tack is taken by Barg et al.~\cite{BFSF13} who study the extension problem for order-preserving bijections on posets.

\medskip

We fix the following notation.
Let $\leq$ be a partial order on $[n]:=\{1,\ldots,n\}$, thus $\bP:=([n],\leq)$ is a poset.
A subset $S\subseteq[n]$ is called an ideal if $i\in S$ and $j\leq i$ implies $j\in S$.
Denote by $\ideal{S}$ the smallest ideal generated by the set~$S$.

As before,~$R$ denotes a finite Frobenius ring.

A poset~$\bP=([n],\leq)$ induces the \emph{poset weight} on~$R^n$  given by
\begin{equation}\label{e-pweight}
   \wt_{\bP}(x)=\big|\ideal{\supp(x)}\big|,
\end{equation}
where, as usual, $\supp(x):=\{i\mid x_i\neq0\}$ denotes the support of~$x=(x_1,\ldots,x_n)\in R^n$.
The weight induces a metric on~$R^n$ (see~\cite[Lem.~1.1]{BGL95} for~$\F^n$).
Both the Hamming weight and the RT-weight are special cases of a poset weight and will be discussed in Example~\ref{E-HierPos} below.
It is easy to see that the induced partition~$\cP_{\bP}$, consisting of the sets
$P_m=\{x\in R^n\mid \wt_{\bP}(x)=m\},\,m=0,\ldots,n$, is in general not induced by a group action on~$R^n$.

In this section we investigate for which posets the $\wt_{\bP}$-preserving linear maps satisfy the
MacWil\-liams extension theorem.

We start with the following simple example which shows that in general the extension theorem does not hold.
It is a simple reformulation of the example in Remark~\ref{R-CumRT}.

%%%%%%%%%%%%%%%%%%%%%%%%%%%%%%%
\begin{exa}\label{E-PosetNonExt}
Let $N=tn$ for some $t,n\in\N$, and consider the poset $\bP=([N],\leq)$ with partial order given by the~$t$ disjoint chains
\[
  1<\ldots<n,\quad n+1<\ldots<2n,\quad 2n+1<\ldots<3n,\ \ldots,\ (t-1)n+1<\ldots<tn.
\]
By the very definition, the poset weight~$\wt_{\bP}$ of a vector $x=(x_1,\ldots,x_t)\in R^N$, where $x_i=(x_{i,1},\ldots,x_{i,n})\in R^n$, is given by
$\wt_{\bP}(x)=\sum_{i=1}^t\max\{s\mid x_{i,s}\neq0\}$.
Thus, $\wt_{\bP}(x)=\wtNRT(M_x)$, where $M_x\in R^{t\times n}$ is the matrix consisting of the rows~$x_i,\,i=1,\ldots,t$, and $\wtNRT$ is the NRT-weight
defined in Remark~\ref{R-CumRT}.
As a consequence, the space $R^N$ endowed with the poset metric induced by~$\bP$ is isometric to the NRT-space
$R^{t\times n}$ discussed in Remark~\ref{R-CumRT}.
Using the above identification, the example given in that remark translates as follows:
the poset is $\bP=([4],\leq)$, where $1<2$ and $3<4$.
In~$\F_2^4$ we consider the codes $\cC=\{(00,00),(10,10)\},\,\cC'=\{(00,00),\,(00,01)\}$.
Since $\ideal{\supp(10,10)}=\{1,3\}$ and $\ideal{\supp(00,01)}=\{3,4\}$, we have $\wt_{\bP}(10,10)=\wt_{\bP}(00,01)=2$, and thus
the unique isomorphism between~$\cC$ and~$\cC'$ is a $\wt_{\bP}$-isometry.
As already shown in Remark~\ref{R-CumRT}, this isometry cannot be extended to an isometry on all of~$\F_2^4$.
As a consequence, the MacWilliams extension theorem does not hold for this particular poset metric.
\end{exa}
%%%%%%%%%%%%%%%%%%%%%%%%%%%%%%

In the proof of Theorem~\ref{T-NonHierPoset} we will generalize the idea of this example to the class of non-hierarchical posets.
More precisely, we will show that the MacWilliams extension theorem holds true if and only if the poset is hierarchical in the following sense.
For the terminology we follow Kim and Oh~\cite{KiOh05}.
%%%%%%%%%%%%%%%%%%%%%%%%%%
\begin{defi}\label{D-hierposet}
The poset $\bP=([n],\leq)$ is called a \emph{hierarchical poset} if there exists a partition
$[n]=\bigcup_{i=1}^t\hspace*{-2em}\raisebox{.5ex}{$\cdot$}\hspace*{1.8em} \Gamma_i$
such that for all $l,m\in[n]$ we have $l< m$ if and only if  $l\in\Gamma_{i_1},\,m\in\Gamma_{i_2}$
for some $i_1<i_2$ (where $i_1<i_2$ refers to the natural order in~$\N$).
In other words, for all $i_1<i_2$ every element in $\Gamma_{i_1}$ is less than every element in $\Gamma_{i_2}$,
and no other two distinct elements in~$[n]$ are comparable.
We call $\Gamma_{i}$ the \emph{$i$-th level} of~$\bP$.
The hierarchical poset is completely determined (up to order-isomorphism)
by the data $(n_1,\ldots,n_t)$, where $n_i=|\Gamma_i|$, and is denoted by $\bH(n;n_1,\ldots,n_t)$.
\end{defi}
%%%%%%%%%%%%%%%%%%%%%%%%%%

The following two familiar examples are the most extreme cases of hierarchical posets.

%%%%%%%%%%%%%%%%%%%%%%%%%%%%%
\begin{exa}\label{E-HierPos}
An \emph{anti-chain} on~$[n]$ is a poset in which any two distinct elements in~$[n]$ are incomparable.
Thus,~$\bP$ is an anti-chain  if and only if~$\bP$ is the hierarchical poset $\bH(n;n)$.
In this case, $\ideal{\supp(x)}=\supp(x)$ for all $x\in R^n$, and $\wt_{\bP}$ is
simply the Hamming weight on~$R^n$.
A poset~$\bP=([n],\leq)$  is a \emph{chain} if~$\leq$ is a total order, in other words, if~$\bP$ is
the hierarchical poset $\bH(n;1,\ldots,1)$.
Assuming without loss of generality that $1<2<\ldots<n$, we observe that
$\wt_{\bP}(x)=\max\{i\mid x_i\neq0\}$, which is exactly the RT-weight~$\wtRT$ considered in~\eqref{e-RTweight}.
\end{exa}
%%%%%%%%%%%%%%%%%%%%%%%%%%%%%%

In preparation for the next theorem we fix the following notation.

Let $\bP$ be the hierarchical poset $\bH(n;n_1,\ldots,n_t)$.
Thus $\sum_{i=1}^t n_i=n$.
In this case it is convenient to write the underlying set~$[n]$ as
\[
       \cN:=\{(i,j)\mid i=1,\ldots,t,\,j=1,\ldots,n_i\}
\]
such that the partial order simply reads as
$(i,j)< (i',j')\Longleftrightarrow i<i'$
(where~$i<i'$ refers to the natural order on~$\N$).
Consequently, for any set $S=\{(i_1,j_1),\ldots,(i_r,j_r)\}\subseteq\cN$, where $i_1\leq\ldots\leq i_{s-1}<i_s=i_{s+1}=\ldots=i_r$,
the ideal generated by~$S$ is
\[
  \ideal{S}=\{(i,j)\in\cN\mid i< i_r\}\cup\{(i_s,j_s),(i_{s+1},j_{s+1}),\ldots,(i_r,j_r)\}.
\]
Accordingly, we write $R^n$ as
\begin{equation}\label{e-Rfactors}
     R^n=R^{n_1}\times\ldots\times R^{n_t},
\end{equation}
so that a vector $x\in R^n$ has the form $x=(x_1,\ldots,x_t)$, where $x_i\in R^{n_i}$.
Now the definition of the poset weight yields
\begin{equation}\label{e-HPweight}
         \wt_{\bP}(x)=\sum_{i=1}^{s-1}n_i+\wt_{\text{H}}(x_s),\text{ where }s=\max\{i\mid x_i\not=0\},
\end{equation}
and where $\wt_{\text{H}}$ stands for the Hamming weight on each of the modules~$R^{n_i}$.

\medskip
Now we are ready to prove the MacWilliams extension theorem for hierarchical posets.
It generalizes both Theorem~\ref{T-rho} and Theorem~\ref{T-MacWExtHamm}.

%%%%%%%%%%%%%%%%%%%%%%%%%%%%%%%%
\begin{theo}\label{T-HierPosMacWExt}
Let $\bP$ be the hierarchical poset $\bH(n;n_1,\ldots,n_t)$ and $\wt_{\bP}$ be the associated poset-weight on $R^n$.
Then every $\wt_{\bP}$-preserving isomorphism $f:\;\cC\longrightarrow\cC'$ between codes $\cC,\,\cC'\subseteq R^n$
is a global $\cL$-map, where the group $\cL\leq \glnr$ consists of the matrices
\begin{equation}\label{e-Mtriang}
    M:=\begin{pmatrix}M_{1,1}& & & \\ M_{2,1}&M_{2,2}& & \\ \vdots&\ddots&\ddots& \\ M_{t,1}&\cdots&M_{t,t-1}&M_{t,t}\end{pmatrix},
\end{equation}
where $M_{i,i}\in\Mon(n_i,R)$ for all $i=1,\ldots,t$ and $M_{i,j}\in R^{n_i\times n_j}$ for $i\neq j$.
As a consequence, $\wt_{\bP}$-preserving maps satisfy the MacWilliams extension theorem.
\end{theo}
%%%%%%%%%%%%%%%%%%%%%%%%%%%%%%%%

Note that the theorem also tells us that~$\cL$ is the group of~$\wt_{\bP}$-isometries on~$R^n$.
This generalizes an earlier result from fields to arbitrary finite Frobenius rings.
Indeed, in~\cite[Cor.~1.3]{PFKH08} Panek et al.\ derive the isometry group for the space $\F^n$ endowed with a poset metric.
For the hierarchical poset of Theorem~\ref{T-HierPosMacWExt}, this group amounts to a semi-direct product $\cG\rtimes\cH$, where
$\cG$ consists of the lower triangular block matrices with matrices from $\Delta(n_i,\F)$ in the $i$-th diagonal block, and
$\cH=S_{n_1}\times\ldots\times S_{n_t}$.
It is straightforward to show that this group is isomorphic to~$\cL$.

\medskip
\begin{proof}
We induct on the number~$t$ of levels of the hierarchical poset.
If $t=1$, then $\wt_{\bP}=\wt_{\text{H}}$ and the result follows from Theorem~\ref{T-MacWExtHamm}.

Let $t\geq2$ and let $f:\,\cC\longrightarrow\cC'$ be a linear $\wt_{\bP}$-preserving map.
Note that~\eqref{e-HPweight} implies $\wt_{\bP}(x)>\sum_{i=1}^{t-1}n_i$ if and only if $x_t\neq0$.
As a consequence,
\begin{equation}\label{e-fxt}
    y_t=0\Longleftrightarrow x_t=0
\end{equation}
whenever $y=f(x)$.
This allows us to consider various derived maps.
Firstly, let $\widetilde{\cC}=\{(x_1,\ldots,x_{t-1})\mid (x_1,\ldots,x_{t-1},0)\in\cC\}$ and likewise
$\widetilde{\cC'}=\{(y_1,\ldots,y_{t-1})\mid (y_1,\ldots,y_{t-1},0)\in\cC'\}$.
Now~\eqref{e-fxt} shows that $f|_{\widetilde{\cC}}$ is a linear $\wt_{\bP}$-preserving map between the codes
$\widetilde{\cC},\,\widetilde{\cC'}\subseteq R^{n_1}\times\ldots\times R^{n_{t-1}}=:R^{\tilde{n}}$.
By induction hypothesis there exists an invertible matrix
\begin{equation}\label{e-Mtilde}
   \tilde{M}=\begin{pmatrix}M_{1,1}& & & \\ M_{2,1}&M_{2,2}& & \\ \vdots&\ddots&\ddots& \\ M_{t-1,1}&\cdots&M_{t-1,t-2}&M_{t-1,t-1}\end{pmatrix},
\end{equation}
with $M_{i,i}\in\Mon(n_i,R)$ for all $i=1,\ldots,t-1$ such that
\[
    f(x_1,\ldots,x_{t-1},0)=\big((x_1,\ldots,x_{t-1})\tilde{M},0\big)\text{ for all }(x_1,\ldots,x_{t-1},0)\in\cC.
\]
If $\tilde{\cC}=\{0\}$, then let~$\tilde{M}$ simply be the identity matrix
(or any other matrix in $R^{\tilde{n}\times\tilde{n}}$ as in~\eqref{e-Mtilde}).

Next consider the projection~$\Pi_t$ of $R^n$ on the last component $R^{n_t}$ and
let $\cC_{(t)}:=\Pi_t(\cC)$ and $\cC'_{(t)}=\Pi_t(\cC')$.
It follows from~\eqref{e-fxt} that the map
\[
  f_{(t)}:\cC_{(t)}\longrightarrow\cC'_{(t)},\ x_t\longmapsto \Pi_t\big(f(x_1,\ldots,x_t)\big), \text{ where }(x_1,\ldots,x_t)\in\cC,
\]
is well-defined. Moreover, it is linear, injective and $\wt_{\text{H}}$-preserving, where the last part follows from~\eqref{e-HPweight}.
Thus by Theorem~\ref{T-MacWExtHamm} there exists a matrix $M_{t,t}\in\Mon(n_t,R)$ such that
\[
    f_{(t)}(x_t)=x_tM_{t,t} \text{ for all }x_t\in \Pi_t(\cC).
\]
It remains to consider the map
\[
  g:\cC\longrightarrow R^{\tilde{n}},\ x\longmapsto \hat{\Pi}\big(f(x)\big),
\]
where~$\hat{\Pi}$ is the projection of~$R^n$ on $R^{\tilde{n}}=R^{n_1}\times\ldots\times R^{n_{t-1}}$.
Note that $f(x)=\big(g(x),\Pi_t(f(x))\big)=\big(g(x),f_{(t)}(x_t)\big)$ for all $x\in\cC$.

Fix $x\in\cC$ and put $z:=(z_1,\ldots,z_{t-1}):=g(x)-(x_1,\ldots,x_{t-1})\tilde{M}$.
Let $y=f(x)$. Then $(y_1,\ldots,y_{t-1})=g(x)$.
We show that each entry of~$z$ is contained in the right ideal $I:=\sum_{j=1}^{n_t}x_{t,j}R$, where $x_{t,j}$ are the entries of
the vector~$x_t$.
In order to do so, let $r\in\text{ann}_l(I)$.
Then $rx_t=0$ and thus $r(x_1,\ldots,x_{t-1})\in\widetilde{\cC}$.
As a consequence, $rg(x)=\hat{\Pi}(f(rx))=r(x_1,\ldots,x_{t-1})\tilde{M}$.
This shows that $rz=0$, and thus $\text{ann}_l(I)\subseteq\text{ann}_l(\tilde{I})$, where $\tilde{I}$ is the right ideal generated by
the entries of~$z$.
Now the double annihilator property of Remark~\ref{R-FrobAnn} implies $\tilde{I}\subseteq I$.
This means that there exists a matrix $B_x\in R^{n_t\times\tilde{n}}$ such that
\[
  g(x)-(x_1,\ldots,x_{t-1})\tilde{M}=x_t B_x.
\]
All of this yields
\[
  f(x)=(x_1,\ldots,x_t)M_x, \text{ where }M_x=\begin{pmatrix}\tilde{M}&0\\B_x&M_{t,t}\end{pmatrix}.
\]
In other words,~$f$ is a local $\cL$-map.
However, only the lower left block of the matrix on the right hand side depends on~$x$.
Therefore, we may proceed as follows.
Define $\hat{M}:=\text{diag}(M_{1,1}^{-1},\ldots,M_{t,t}^{-1})$ and $\cC'':=\cC'\hat{M}$ along with
the map $f':\cC\longrightarrow\cC'',\ x\longmapsto (x_1,\ldots,x_t)M_x\hat{M}$.
Note that the matrix $M_x\hat{M}$ is of the form as in~\eqref{e-Mtriang}, but with identity matrices on the diagonal.
As a consequence, the map~$f'$ is a local $\cU(S)$-map, where~$S$ is the matrix ring consisting of all matrices of the form
\[
  \begin{pmatrix}D_{1,1}& & & \\ F_{2,1}&D_{2,2}& & \\ \vdots&\ddots&\ddots& \\ F_{t,1}&\cdots&F_{t,t-1}&D_{t,t}\end{pmatrix},
\]
where $F_{i,j}\in R^{n_i\times n_j}$ and $D_{i,i}$ are diagonal matrices in $R^{n_i\times n_i}$.
Since the diagonal matrices in $R^{n_i\times n_i}$ form a constructible ring, Theorem~\ref{T-triangular}(b) tells us that the
ring~$S$ is constructible.
Thus by Theorem~\ref{T-Constr}, the map~$f'$ is a global $\cU(S)$-map, and hence a global $\cL$-map.
But then the map~$f$, which differs from~$f'$ only by the global matrix~$\hat{M}\in\cL$, is a global~$\cL$-map as well.
This concludes the proof.
\end{proof}

%%%%%%%%%%%%%%%%%%%%%%%%%%%%%
\begin{exa}\label{E-HierPosExt}
In $\F_2^4$ we consider the poset weight with respect to the hierarchical poset $\bP:=\bH(4;2,2)$.
Thus $\wt_{\bP}(x_1,x_2,x_3,x_4)=2+\wtH(x_3,x_4)$ if $(x_3,x_4)\neq0$ and $\wt_{\bP}(x_1,x_2,x_3,x_4)=\wtH(x_1,x_2)$ otherwise.
Let $\cC=\inner{0000,\,1010,\,0111,\,1101}$ and $\cC'=\inner{0000,\,1110,\,1111,\,0001}$.
Then the map~$f:\cC\longrightarrow\cC'$ that maps the vectors of~$\cC$ to those of~$\cC'$ in the given order is linear and
$\wt_{\bP}$-preserving.
With a  routine computation one can verify that~$f$ can be extended to a global $\wt_{\bP}$-isomorphism in exactly two ways, given by the matrices
\[
   M_1:=\begin{pmatrix}0&1&0&0\\1&0&0&0\\1&0&1&0\\1&1&0&1\end{pmatrix},\
   M_2:=\begin{pmatrix}1&0&0&0\\0&1&0&0\\0&1&1&0\\1&1&0&1\end{pmatrix}.
\]
Note that in this case, the code $\tilde{\cC}=\{(x_1,x_2,x_3,x_4)\in\cC\mid x_3=0=x_4\}$ is trivial, and the matrix~$\tilde{M}$
in~\eqref{e-Mtilde}, which is in ${\rm Mon}(2,\F_2)$, may be chosen in two different ways.
\end{exa}
%%%%%%%%%%%%%%%%%%%%%%%%%%%%%

Our final result may be regarded as the converse of Theorem~\ref{T-HierPosMacWExt}.
%%%%%%%%%%%%%%%%%%%%%%%%%%%%%%
\begin{theo}\label{T-NonHierPoset}
Let $\bP=([n],\leq)$ be a poset such that the MacWilliams extension theorem holds for $\wt_{\bP}$-preserving isomorphisms.
Then~$\bP$ is hierarchical.
More precisely, if~$\bP$ is not hierarchical then there exists a pair of codes $\cC,\,\cC'\subseteq R^n$ and a $\wt_{\bP}$-isometry
$f:\cC\longrightarrow\cC'$ that cannot be extended to a $\wt_{\bP}$-isometry on~$R^n$.
\end{theo}
%%%%%%%%%%%%%%%%%%%%%%%%%%%%%%

For the proof we will make use of $\min(\bP)$, which denotes the set of all minimal elements of~$\bP$.
Moreover, for a subset~$S$ of~$[n]$, we denote by $\bP\backslash S$ the poset $([n]\backslash S,\leq)$.

\begin{proof}
Let~$\bP$ be a non-hierarchical poset.
Consider the level sets
\[
  \Gamma^{(1)}:=\min(\bP),\ \Gamma^{(2)}=\min(\bP\backslash\Gamma^{(1)}),\ \Gamma^{(3)}=\min\big(\bP\backslash(\Gamma^{(1)}\cup\Gamma^{(2)})\big),\ldots.
\]
Then $[n]=\bigcup_{\ell=1}^L\hspace*{-2.1em}\raisebox{.5ex}{$\cdot$}\hspace*{1.8em} \Gamma^{(\ell)}$ for some $L\in\N$.
Moreover, elements in the same level set are incomparable.
If~$\bP$ is hierarchical, then the level sets~$\Gamma^{(\ell)}$ are exactly the sets~$\Gamma_\ell$ of Definition~\ref{D-hierposet}.
In that case we have for any $i,j\in[n]$ that $i<j$ if and only if there exist $\ell<m$ (w.r.t.\ the natural order in~$\N$)
such that $i\in\Gamma^{(\ell)}$ and $j\in\Gamma^{(m)}$.

Thus, since~$\bP$ is not hierarchical, there exists a minimal~$\ell$ and elements $\alpha\in\Gamma^{(\ell)}$ and $\beta\in\Gamma^{(\ell+1)}$ such
that $\alpha\not<\beta$.
Put $\tilde{\Gamma}:=\bigcup_{r=1}^{\ell-1} \Gamma^{(r)}$.
Then minimality of~$\ell$ implies
\begin{equation}\label{e-jGamma}
  j\in\tilde{\Gamma},\ i\in\Gamma^{(\ell)}\Longrightarrow j<i.
\end{equation}
Define $B:=\{i\in\Gamma^{(\ell)}\mid i<\beta\}$ and $B':=B\cup\{\alpha\}$. Note that $\alpha\not\in B$.
The definition of the level sets~$\Gamma^{(r)}$ yields
\begin{equation}\label{e-supps}
  \ideal{\beta}=\{\beta\}\cup\{j\mid j\leq i\text{ for some }i\in B\}=\{\beta\}\cup B\cup \tilde{\Gamma}\ \text{ and }\
  \ideal{B'}=B'\cup \tilde{\Gamma}.
\end{equation}
Now we are ready to define suitable isometric codes in~$R^n$.
Denote by $e_1,\ldots,e_n$ the standard basis vectors of~$R^n$ and put $\hat{e}:=\sum_{i\in B'}e_i$.
Then $\wt_{\bP}(e_{\beta})=\wt_{\bP}(\hat{e})=|B|+|\tilde{\Gamma}|+1$ due to~\eqref{e-supps}.
Let~$\cC,\,\cC'\subseteq R^n$ be the one-dimensional codes generated by~$\hat{e}$ and~$e_{\beta}$, respectively, and define
the linear map~$f:\cC\longrightarrow\cC'$  by $f(\hat{e})=e_{\beta}$.
Then~$f$ is a $\wt_{\bP}$-isometry between~$\cC$ and~$\cC'$.

It remains to show that~$f$ cannot be extended to a $\wt_{\bP}$-isometry on~$R^n$.
In order to do so assume to the contrary that there exists such an extension, say~$\hat{f}$.
We argue as follows.
Let $x\in R^n$. If $\supp(x)$ contains an element $i\in[n]\,\backslash\,(\cup_{r=1}^\ell\Gamma^{(r)})$, then
the ideal $\ideal{\supp(x)}$ also contains an element in~$\Gamma^{(\ell)}$ due to the definition of the level sets.
Thus with the aid of~\eqref{e-jGamma} we obtain
\[
   \supp(x)\cap\Big([n]\,\backslash\,(\bigcup_{r=1}^\ell\Gamma^{(r)})\Big)\not=\emptyset\Longrightarrow
   \wt_{\bP}(x)\geq |\tilde{\Gamma}|+2.
\]
On the other hand, if $\supp(x)$ is in $\cup_{r=1}^\ell\Gamma^{(r)}$, then
\[
  \wt_{\bP}(x)= \Big|\bigcup_{r=1}^{m-1}\Gamma^{(r)}\Big|+\big|\supp(x)\cap\Gamma^{(m)}\big|, \text{ where $m\leq\ell$ is maximal such that }
  \supp(x)\cap\Gamma^{(m)}\not=\emptyset.
\]
All of this shows that the set of vectors with weight $\big|\bigcup_{r=1}^{m-1}\Gamma^{(r)}\big|+1,\,m\leq\ell$, is given by
$A_m:=\{\gamma e_i +v\mid \gamma\in R\backslash\{0\},i\in\Gamma^{(m)},\, \supp(v)\subseteq\cup_{r=1}^{m-1}\Gamma^{(r)}\}$.
As a consequence,~$\hat{f}(A_m)=A_m$ for all $m\leq\ell$.
Using that~$B'\subseteq\Gamma^{(\ell)}$, this in turn implies that
$\hat{f}(\hat{e})=\sum_{i\in B'}\hat{f}(e_i)=\sum_{i\in B''}\gamma_ie_i+v$ for some set
$B''\in\Gamma^{(\ell)}$, some $\gamma_i\in R\backslash\{0\}$ and $v\in R^n$ such that $\supp(v)\subseteq\tilde{\Gamma}$.
But this contradicts the definition of~$f$.
This concludes the proof.
\end{proof}

For $R=\F$ being a field, the non-existence of an extension~$\hat{f}$ can
also be shown with the aid~\cite[Cor.~1.3]{PFKH08}, where the group of isometries of~$(\F^n,\wt_{\bP})$ has
been derived.
In our specific situation, however, the above direct reasoning leads immediately to the desired result.

\section*{Acknowledgement}
We would like to thank Jay Wood for his valuable comments.
We are also grateful to him for encouraging us to treat the non-commutative case as well.

%%%%%%%%%%%%%%%%%%%%%%%%%%%%%%%%%%%%%%%%%%%%%%
\bibliographystyle{abbrv}

\end{document}